\documentclass[12pt]{iopart}

\usepackage{graphics,graphicx,amssymb,amsthm}

\usepackage{color}

\def\xsum{\mathop{\sum\nolimits'}}
\newtheorem{thm}{Theorem}[section]
\newtheorem{defi}{Definition}[section]
\newtheorem{corollary}[thm]{Corollary}
\newtheorem{prop}[thm]{Proposition}

\theoremstyle{definition}
\newtheorem{remark}[thm]{Remark}

\newtheorem{example}[thm]{Example}

\newcommand{\R}{\mathbb R}

\newcommand{\Z}{\mathbb Z}
\newcommand{\N}{\mathbb N}

\begin{document}

\title[Equidistant versus bipartite ground states for 1D fluids]{Equidistant versus bipartite ground states for 1D classical fluids at fixed particle density}

\author{Laurent B\'etermin}

\address{Institute Camille Jordan, Universit\'e Claude Bernard Lyon 1,
Villeurbanne, France}

\vskip0.3truecm

\author{Ladislav \v{S}amaj and Igor Trav\v{e}nec}

\address{Institute of Physics, Slovak Academy of Sciences, 
D\'ubravsk\'a cesta 9, 84511 Bratislava, Slovakia}
\ead{fyzitrav@savba.sk}

\vspace{10pt}
\begin{indented}
\item[]
\end{indented}

\begin{abstract}
We study the ground-state properties of one-dimensional fluids of
classical (i.e., non-quantum) particles interacting pairwisely via
a potential, at the fixed particle density $\rho$.
Restricting ourselves to periodic configurations of particles,
two possibilities are considered: an equidistant chain of particles with
the uniform spacing $A=1/\rho$ and its simplest non-Bravais modulation,
namely a bipartite lattice composed of two equidistant chains,
shifted with respect to one another.
Assuming the long range of the interaction potential, 
the equidistant chain dominates if $A$ is small enough, $0<A<A_c$.
At a critical value of $A=A_c$, the system undergoes a continuous
second-order phase transition from the equidistant chain to a bipartite lattice.
The energy and the order parameter are singular functions
of the deviation from the critical point $A-A_c$ with universal
(i.e., independent of the model's parameters) mean-field values of
critical exponents.
The tricritical point at which the curve of continuous second-order
transitions meets with the one of discontinuous first-order transitions
is determined.
The general theory is applied to the Lennard-Jones model with the $(n,m)$
Mie potential for which the phase diagram is constructed.
The inclusion of a hard-core around each particle reveals a non-universal
critical phenomenon with an $m$-dependent critical exponent.
\end{abstract}


\vspace{2pc}


\maketitle

\renewcommand{\theequation}{1.\arabic{equation}}
\setcounter{equation}{0}

\section{Introduction} \label{Sec1}
In quantum mechanics, a Peierls transition \cite{Peierls91} is an instability
of the ground state (i.e., the state/configuration of minimal energy)
of a one-dimensional (1D) equally-spaced chain of ions,
with one electron per ion.
Instead of the equidistant chain, its bipartite (i.e., two equidistant
sublattices shifted with respect to one another) distortion becomes
energetically favorable: the lowering of the energy due to the modified
electron band gap outweighs the elastic energy cost of the rearranging of ions.
Our aim is to show that a similar phase transition can occur also in
the ground state of 1D fluids of classical (i.e., non-quantum) particles
interacting via pairwise potentials.  

The theory of structural phase transitions is developed
in this paper for general 1D fluids.
Its particular application is made to the typical interaction potential at
the atomic level for two particles at distance $r$,
the Lennard-Jones (LJ) one \cite{Lennard31}
\begin{equation} \label{LJ}
\forall r>0,\quad f(r) = \frac{1}{r^{12}} - \frac{2}{r^6} , 
\end{equation}
where the repulsive term $1/r^{12}$ (the repulsive exponent $12$ was
originally chosen for computational reasons) is dominant at small
distances and the attractive Van der Waals term $-2/r^6$ is dominant
at large distances.
This potential is a special case of the general $(n,m)$ Mie potential
\cite{Mie03}
\begin{equation} \label{Mie}
\forall r>0,\quad f_{nm}(r) = \frac{1}{n-m} \left(
\frac{m}{r^n}-\frac{n}{r^m} \right), \qquad n>m>1. 
\end{equation}
Here, the normalization is chosen such that the potential minimum is at
distance $r_{\min}=1$, which implies $f(r_{\min})=-1$; thus the length and energy
can be taken in arbitrary units.
The condition $n>m$ ensures that the repulsive term of order $1/r^n$
is dominant at small distances $r$ and the attractive term of order $1/r^m$
is dominant at large $r$ in such a way that $f_{nm}$ is a one-well potential.
The condition $m>1$ ensures the lattice summability of the energy
per site (i.e., the integrability of the potential at infinity).
Note that the limit $n\to m^+$ of the potential (\ref{Mie}) is well defined:
\begin{equation} \label{Miemm}
\forall r>0,\quad \lim_{n\to m^+} f_{nm}(r) = - \frac{1+m\ln r}{r^m} .
\end{equation}  
To avoid confusion with the contemporary physical literature, we shall
refer to the model of particles interacting via Mie potential to as
the (generalized) LJ model.

Proving rigorously that the ground state of a system of classical
particles with pairwise interactions is a periodic crystal turns out to be
a difficult task.
In one dimension (1D), it was proven for the LJ and similar
potentials, that the {\em global} ground state is unique and corresponds to
an equidistant array of particles in the infinite-particle limit
\cite{Gardner79,Hamrick79,Katz84,Stillinger95,Ventevogel78}.
Note that the 1D LJ model was studied, mostly numerically,
also for non-zero temperatures \cite{Lee17,Lepri05}.
In two dimensions (2D), the proof that the ground state is periodic
was accomplished for the ``idealized'' pair potential
\begin{equation}
f_{\rm id}(r) = \left\{
\begin{array}{ll}
+\infty & \mbox{if $r<1$,} \cr
-1 & \mbox{if $r=1$,} \cr
0 & \mbox{if $r>1$,}
\end{array}  
\right.  
\end{equation}
first by Heitmann and Radin in \cite{Rad2} using graph theory methods,
and recently by De Luca and Friesecke \cite{Luca17} with the aid of
a discrete Gauss-Bonnet theorem.
Let us add that the first author and Furlanetto recently gave in
\cite{BetFur24} a complete characterization of minimizers for
this problem with respect to the chosen norm on $\R^2$.
No rigorous proof about crystallization for radially symmetric pair
potential is known in three dimensions.

The crystallization problem has been reviewed in
\cite{BetDelPet,Blanc15,Bris05,Canizo24,Radin87}.
The ``crystallization conjecture'' states that most of the non-trivial
ground states of interacting systems correspond to periodic lattices,
reducing in this way the number of degrees of freedom of the general
crystallization problem.
Restricting oneself to periodic (or even Bravais) lattices,
one considers usually the energy minimization
{\em at a fixed particle density} $\rho$, in the spirit of the grandcanonical
ensemble with the chemical potential coupled to the microscopic particle number.
In practice, the particles are confined to a large domain $\Omega$ of volume
$\vert\Omega\vert$, say a torus which covers periodic configurations and
lattices, and one imposes their number $N$ to be $N=\rho\vert \Omega\vert$.
Only very few results are available in this field.
In the simplest $N=1$ case, the triangular lattice has been shown to be universally optimal (among lattices), i.e. for
any potential $f$ where $f(r)=F(r^2)$ and $F$ is a completely monotone interaction potential, in 2D \cite{Cohn07,Montgomery88},
where the same result still holds at fixed high density when the pair
potential is equivalent to a completely monotone one near the origin
\cite{Betermin15}.
Optimal lattices for non-completely monotone interaction potentials
were studied for instance in \cite{Betermin19,SunWeiZou24}.
The 2D LJ system is well understood at present.
Furthermore, Luo and Wei \cite{LuoWeiLJ22} proved that the hexagonal, rhombic,
square and rectangular structures minimize successively the interaction energy
as the particle density decreases, as conjectured by the present authors
in \cite{Betermin18,Travenec19}.
The extension to non-Bravais ``zig-zag'' structures with two, three, etc.
particles per unit cell leads to even lower interaction energies at low density
\cite{Travenec19}.
In higher dimension, it has to be noticed that there is no universal optimizer
at the fixed density in 3D \cite{Betermin17,SarStromb} and all our knowledge
is based on numerical calculations and (non-)local optimality results only \cite{Betermin23}.
Furthermore, it was proven recently that the Gosset lattice ${\bf E}_8$ and
the Leech lattice $\Lambda_{24}$ are universally optimal among all periodic
structures in dimensions 8 and 24, respectively \cite{Cohn22}.

In what follows, we restrict our list of references to the ground state of
1D fluids of classical particles interacting via a pair potential $f$,
at the fixed mean particle density $\rho$.
The crucial question in this field is whether the ground state of the
particle system is the crystal and in the affirmative case whether it
corresponds to an equidistant array of particles with the spacing $A=1/\rho$
or to a modulated structure with a unit cell involving more than one particle.
Ventevogel \cite{Ventevogel78} has proved a number of theorems about
the ground state of 1D infinite particle systems interacting with various
types of purely repulsive pairwise potentials $f$. 
In particular, he proved an equidistant ground state for any density $\rho>0$
in the case of {\em convex} repulsive potentials.
In subsequent works \cite{Ventevogel79a,Ventevogel79b}, the same result
have been obtained for the {\em non-convex} repulsive potentials
$x\mapsto\exp(-\alpha x^2)$ $(\alpha>0)$ and
$x\mapsto 1/(b^2+x^2)^n$ $(n>\frac{1}{2}, b\in \mathbb{R})$.
Repulsive potentials $f$, for which a modulated structure has a lower
energy than the equidistant one, were studied in \cite{Nijboer85}.
For some potentials like $x\mapsto 1/(1+x^4)$, this happens at certain
particle densities.
But as soon as the interaction is repulsive at small distances $x$ and
attractive at large $x$, interaction potentials which lead to modulated
incommensurate ground states are not pairwise as was shown by
by Janssen and Tjon \cite{Janssen82,Janssen83}.

In this paper, restricting ourselves to periodic ground-state configurations,
two possibilities are considered: an equidistant chain of particles with
the spacing $A=1/\rho$ and its simplest modulation, namely a bipartite
lattice composed of two equidistant (with spacing $2A$) chains
shifted with respect to one another.
The method is based on a series expansion of the energy in
integer powers of a small parameter of the bipartite chain which determines
exactly the critical point of a continuous second-order transition.
Assuming a long-range decay of the interaction potential $f$,
the equidistant chain dominates if $A$ is small enough, $0<A<A_c$.
At $A=A_c$, the system undergoes a continuous second-order phase transition
from the equidistant Bravais chain (with one particle per unit cell)
to a bipartite lattice with two particles per unit cell.
Increasing $A$ beyond $A_c$, $A>A_c$, the ground state of the system
corresponds to bipartite lattices with continuously changing shift. 
Close to $A_c$, $A\to A_c^+$, the shift between the two sublattices behaves
as $(A-A_c)^{\beta}$ with the (universal) mean-field value of
the critical exponent $\beta=\frac{1}{2}$. 
The conditions for a tricritical point at which second-order and first-order
phase transitions meet \cite{Griffiths70,Landau37,Landau80} is determined
exactly as well. 

The general theory is applied to the 1D $(n,m)$ LJ model (\ref{Mie}).
The energy per particle for the equidistant and bipartite configurations
is expressed in terms of the Riemann and Hurwitz zeta functions.
The $n\to\infty$ limit of the $(n,m)$ LJ potential (\ref{Mie})  
corresponds to the hard-core (hc) potential of type $f(r)\to\infty$
for $0<r<1$ and $f_m(r) = -1/r^m$ for $r\ge 1$.
The hc potential reflects adequately strong repulsion between particles
at small distances.
This motivates us to include in our study a hc extension of the LJ potential,
namely
\begin{equation} \label{fhc}
f_{nm}^{\rm hc}(r) = \left\{
\begin{array}{ll}
+\infty & \mbox{if $r<\sigma$} \cr
f_{nm}(r) & \mbox{if $r\ge \sigma$.}
\end{array}  
\right.  
\end{equation}
Here, the hard-core radius $\sigma$ is a new length which influences the nature
of the ground state.
Certain aspects of this model turn out to be non-universal with
an exponent dependent exclusively on the LJ parameter $m$.

\medskip

Let us now summarize our main results concerning the general bipartite configuration model and the 1D $(n,m)$ Lennard-Jones case in the following theorem. Details and proofs can be found throughout the paper.
\begin{thm}[\textbf{Summary of the main results}]
Let $f:\R\to \R$, integrable at infinity, be such that there exists a Radon measure $\mu_f$ on $(0,+\infty)$ satisfying
$$
\forall x\in \R,\quad f(x)=\int_0^\infty e^{-x^2 t}d\mu_f(t).
$$
Then, the $f$-energy of the bipartite chain parametrized by $\Delta>0$ and $A>0$ and given by
$$
L_{A,\Delta}:=2A\Z \cup \left( 2A\Z + \frac{2A\Delta}{1+\Delta}\right)
$$
is
\begin{eqnarray} 
E^{\rm bip}(A,\Delta) & = & \frac{1}{2} \int_0^{\infty} \left(
\sum_{j=-\infty}^{\infty} e^{-4j^2 A^2 t} - 1 \right) d\mu_f(t) \nonumber \\
& & + \frac{1}{4} \int_0^{\infty} \sum_{j=-\infty}^{\infty}
\left[ e^{-(a+2jA)^2 t} + e^{-(b+2jA)^2 t} \right] d\mu_f(t) , \label{biprepr2intro}
\end{eqnarray}
where the periods $a$ and $b$ are defined as
\begin{equation} \label{periodsintro}
a=a(A,\Delta) := \frac{2A\Delta}{1+\Delta} ,
\qquad b=b(A,\Delta) := \frac{2A}{1+\Delta} . 
\end{equation}
\textbf{• Lennard-Jones model.} Let $A>0$, $f=f_{nm}$, $n>m>1$, defined by (\ref{Mie}), then, as $\varepsilon\to 0$, we have
\begin{equation} \label{trans}
E_{nm}^{\rm bip}\left( A,e^{\varepsilon}\right) = E_{nm}^{\rm eq}(A)
+ E_2^{nm}(A) \varepsilon^2 + E_4^{nm}(A) \varepsilon^4 + O(\varepsilon^6), 
\end{equation}
where the energy of the equidistant configuration is given by
\begin{equation} \label{energd1intro}
E_{nm}^{\rm eq}(A) =E_{nm}^{\rm bip}\left( A,1\right)= \frac{1}{n-m} \left[ \frac{m\, \zeta(n)}{A^n}
-\frac{n\ \zeta(m)}{A^m} \right],
\end{equation}
where $\zeta$ is the Riemann zeta function, and the coefficients are given in terms of $\zeta$  by (\ref{rov1}) and (\ref{rov2}) below.
\begin{enumerate}
\item[-] \textbf{Transition point and its asymptotics.} The unique transition point $A_c^{nm}$, i.e. such that $E_2(A_c^{nm})=0$ with a change of sign at $A=A_c^{nm}$, is given by
\begin{equation} \label{astar}
A_c^{nm} = \frac{1}{2} 
\left[ \frac{(2^{2+n}-1)(1+n)\zeta(n+2)}{(2^{2+m}-1)(1+m)\zeta(m+2)}
\right]^{\frac{1}{n-m}} .  
\end{equation} 
If $A<A_c^{nm}$, then the equidistant configuration $L_{A,1}$ corresponding to $\Delta=1$ is the unique minimizer of $\Delta\mapsto E_{nm}^{\rm bip}\left( A,\Delta\right)$ while this minimizer is a bipartite configuration (i.e. $\Delta\neq 1$) when $A>A_c^{nm}$.\\
Moreover, in the limit $n\to m^+$, this relation yields
\begin{equation}
A_c^{m^+,m}:=\lim_{n\to m^+} A_c^{nm}=\exp\left( \frac{\ln 2}{2^{2+m}-1}+\frac{1}{1+m}
+\frac{\zeta'(m+2)}{\zeta(m+2)} \right),
\end{equation}
and in the hard-core limit $n\to\infty$, the critical inverse particle density
tends to the $m$-independent value
\begin{equation}
\lim_{n\to\infty} A_c^{nm} = 1.
\end{equation}
\item[-] \textbf{Absence of tricritical point.} Furthermore, there is no tricritical solution $A^t>0$ to the equations
$E_2^{nm}(A^t)=E_4^{nm}(A^t)=0$, i.e. all phase transitions in
the 1D $(n,m)$ LJ model are of second order.
\item[-] \textbf{Asymptotics in the bipartite phase.} The region of the bipartite phase is defined by $A>A_c^{nm}$:  
\begin{itemize}
\item[(1)]
for a given $A>A_c^{nm}$, there exists a nontrivial value of the parameter
$\Delta(A)\ne 1$ given by the equation
\begin{equation} \label{adelintro}
A^{m-n}=\frac{2^{n-m}[\zeta(m+1,1-\delta)-\zeta(m+1,\delta)]}{
\zeta(n+1,1-\delta)-\zeta(n+1,\delta)},\quad \delta=\frac{1}{1+\Delta}.
\end{equation}
\item[(2)] as $A\to +\infty$, we have
\begin{equation} \label{deltaasintro}
\Delta = 2A - 1 + o(1)
\end{equation}
regardless of the values of the LJ parameters $(n,m)$.
\end{itemize}
\end{enumerate}
\end{thm}
\medskip

\textbf{Plan of the paper.} The article is organized as follows.
In section \ref{Sec2}, a general formalism is developed for a phase transition
from the equidistant to bipartite ground states for 1D systems of particles
interacting pairwisely.
The application of the general method to the 1D $(n,m)$ LJ model with
the Mie potential (\ref{Mie}) is presented in section \ref{Sec3}. 
The phase diagram is constructed for various values of the $(n,m)$ parameters.
The inclusion of the hard-core with radius $\sigma$ and the corresponding
non-universal behaviour of the particle system are discussed in section
\ref{Sec4}.

\renewcommand{\theequation}{2.\arabic{equation}}
\setcounter{equation}{0}

\section{General formalism} \label{Sec2}
In this section, we present the general formalism for 1D fluids of classical
particles interacting pairwisely via a potential.
This includes the studied types of interaction potentials, the two types
of periodic ground-state configurations of particles considered in this
paper and the corresponding explicit formulas for the energy per particle.

\subsection{Interaction potentials and general energy} \label{Sec20}

As is usual in the literature, basically since the work of Cohn and Kumar
\cite{Cohn07}, we consider interaction potentials
$f:\mathbb{R}\to \mathbb{R}$ which can be written as $f(x)=g(x^2)$
where $g$ -- which is integrable at infinity -- is the Laplace transform of
a Radon measure $\mu_f$ on $(0,+\infty)$.
\begin{defi}[\textbf{Admissible potentials}]
Let $\varepsilon>0$. We say that $f:\R\to \R$ is admissible, and we write $f\in \mathcal{F}$
if $f(x)=O(|x|^{-1-\varepsilon})$, as $|x|\to +\infty$, and if there exists a Radon measure $\mu_f$ on $(0,+\infty)$ such that, for all $ x\in \R^*$,
\begin{equation} \label{intpot}
f(x) = \int_0^{\infty} e^{-x^2 t} d\mu_f(t).
\end{equation}
\end{defi}
\begin{remark}[\textbf{Symmetry}]
The general interaction potential of two particles with coordinates
$x_1\in \R$ and $x_2\in \R$, $f(x_1-x_2)$, depends on the distance
$\vert x_1-x_2\vert$ only, so it exhibits the symmetry $f(x)=f(-x)$.
\end{remark}

\begin{example}[\textbf{The Riesz potential}]
For the Riesz potential $1/r^s$, $s>1$, which is relevant in this work,
it holds that
\begin{equation} 
f(x) = \frac{1}{\vert x\vert^s} , \qquad  d\mu_f(t)
= \frac{1}{\Gamma(s/2)} t^{s/2-1} dt  
\end{equation}
with $\Gamma$ being the Euler Gamma function.\\
Notice that for $s>1$, this potential is admissible. It is however possible to define an energy, in the sense of the next definition, when $s<1$, using the analytic continuation of the zeta functions involved in the formulas, see Remark \ref{rmk_analytic}.
\end{example}
\begin{defi}[\textbf{General $f$-energy of a periodic configuration}]
Let $f\in \mathcal{F}$, $N\in \N\backslash \{0\}$, $I\subset \Z$ finite interval of
integers and $X=\{x_j\}_{j\in \Z}\subset \R$ be an infinite configuration
which is $N$-periodic in the following sense
$$
\forall j\in \Z, \quad \exists ! (i,k)\in I\times \Z,\quad x_j=x_i+k N .
$$
Then the general $f$-energy of $X$ is defined by
$$
E(X)=\frac{1}{2\sharp I}\sum_{i\in I} \sum_{j\in \Z\backslash \{i\}} f(x_i-x_j),
$$
where the prefactor $\frac{1}{2}$ is present because each interaction energy
is shared by a pair of particles.
\end{defi}

\begin{remark}
In this paper, we will sometimes use the compact notation $\displaystyle \sum_{j\in \Z\backslash \{0\}}=\displaystyle \xsum_{j\in \Z}$.
\end{remark}

\subsection{Equidistant configurations} \label{Sec21}

Let us first consider an infinite equidistant array of particles .
\begin{prop}[\textbf{Energy of an equidistant configuration}]
Let $I=\{0\}$ and  $X_A:=\{j A\}_{j\in \Z}$, where $A$ the lattice spacing (or the inverse
particle density $1/\rho=A$), then
\begin{equation} \label{eqenergy}
E^{\rm eq}(A) :=E(X_A)= \frac{1}{2} \xsum_{j\in \Z} f(\vert j\vert A) 
= \frac{1}{2} \int_0^{\infty} \left[ \theta_3\left( e^{-A^2 t}\right) - 1\right]
d\mu_f(t) ,
\end{equation}
where
\begin{equation}
\theta_3(q) = \sum_{j\in \Z} q^{j^2}
\end{equation}
denotes a Jacobi theta function with zero argument.
\end{prop}
\begin{remark}
For further information about theta functions, see \cite{Gradshteyn}.
\end{remark}
\begin{proof}
We have (for instance) $I=\{0\}$ and therefore
$$
E(X_A)=\frac{1}{2} \sum_{j\in \Z\backslash \{0\}} f(x_0-x_j)
=\frac{1}{2} \xsum_{j\in \Z} f( j A) =\frac{1}{2} \xsum_{j\in \Z} f(\vert j\vert A).
$$
Since $f$ is integrable at infinity, it follows that, interchanging the sum and the integral by absolute summability, 
$$
E(X_A)=\frac{1}{2} \xsum_{j\in \Z} f(\vert j\vert A)=\frac{1}{2} \xsum_{j\in\Z}
\int_0^{+\infty} e^{-j^2 A^2 t} d\mu_f(t)=\frac{1}{2}
\int_0^{+\infty} \left(\xsum_{j\in \Z} e^{-j^2 A^2 t}\right) d\mu_f(t).
$$
Since
$\displaystyle\xsum_{j\in \Z} e^{-j^2 A^2 t}=\theta_3\left( e^{-A^2 t}\right)-1$,
we get the desired formula.
\end{proof}
\begin{remark}{\textbf{Behaviour of $\mu_f$ near the origin.}}
With regard to the Poisson summation formula \cite{Stein71}
\begin{equation}\label{eq_thetaPoisson}
\sum_{j\in \Z} e^{-(j+\phi)^2t} = \sqrt{\frac{\pi}{t}}
\sum_{j\in \Z} e^{2\pi{\rm i}j\phi} e^{-(\pi j)^2/t}
\end{equation}
taken with $\phi=0$ and $t\to A^2 t$, the Jacobi theta function
$\theta_3\left( e^{-A^2 t}\right)$ behaves in the limit $t\to 0$ as, from the right-hand-side of (\ref{eq_thetaPoisson}),
\begin{equation}
\theta_3\left( e^{-A^2 t}\right) =
\sqrt{\frac{\pi}{A^2 t}} + 2\sqrt{\frac{\pi}{A^2 t}}
\exp\left( - \frac{\pi^2}{A^2 t} \right) + o\left( \frac{1}{\sqrt{t}}
\exp\left( - \frac{\pi^2}{A^2 t} \right)\right)
\end{equation}
and in the limit $t\to \infty$ as, from the left-hand-side of (\ref{eq_thetaPoisson}),
\begin{equation}
\theta_3\left( e^{-A^2 t}\right) =
1 + 2 \exp\left( - A^2 t \right) + o\left(\exp\left( - A^2 t \right) \right).
\end{equation}
Consequently, the integral for the energy in (\ref{eqenergy}) converges
provided that the Radon measure is of the order, as $t\to 0^+$,
\begin{equation}
d\mu_f(t) = o\left( \frac{1}{\sqrt{t}} \right) .
\end{equation}  
\end{remark}

In dependence on the Radon measure $\mu_f$ (i.e., on the interaction
potential $f$), the energy (\ref{eqenergy}) as a function of $A$ can exhibit
a nonmonotonous behavior, with only one minimum at a specific $A_{\min}$.
In particular, we have the following characterization.

\begin{prop}[\textbf{Sufficient condition for an equidistant local minimizer}]
Let $f\in \mathcal{F}$ and $A_{\min}>0$. If
\begin{equation} \label{Amin2}
\int_0^{\infty} t \theta_3^{(1)}\left( e^{-A_{\min}^2 t}\right)
d\mu_f(t) = 0 , \nonumber \\ \int_0^{\infty} t^2
\theta_3^{(2)}\left( e^{-A_{\min}^2 t}\right) d\mu_f(t) > 0 , 
\end{equation}  
then $A\mapsto E^{\rm eq}(A)$ admits a strict local minimum at $A=A_{\min}$,
where
\begin{equation}
\theta_3^{(n)}\left( e^{-A^2t}\right):=
\frac{\partial^n}{\partial t^n} \theta_3\left( e^{-A^2 t} \right)
= (-1)^n A^{2n} \sum_{j=-\infty}^{\infty} j^{2n} e^{-j^2 A^2 t}. 
\end{equation}
\end{prop}
\begin{proof}
It is clear that this local minimality occurs if
\begin{equation} \label{Amin1}
\frac{\partial E^{\rm eq}(A)}{\partial A} \Bigg\vert_{A=A_{\min}} = 0 , \qquad
\frac{\partial^2 E^{\rm eq}(A)}{\partial A^2} \Bigg\vert_{A=A_{\min}} > 0,
\end{equation}
which is equivalent to (\ref{Amin2}) by absolute summability and
differentiation under the integral.
\end{proof}

\subsection{Bipartite chain} \label{Sec22}
Next we consider that the particles form a bipartite chain where
the distance between nearest neighbors alternates as (up to a rescaling
multiplicative constant) 1, $\Delta$, 1, $\Delta$, etc. with
$\Delta$ being a free parameter of the bipartite chain \cite{Betermin23b}. 

\begin{defi}[\textbf{Bipartite chain}]
Let $\Delta>0$ and $A>0$.
Then we define the following 1D periodic configuration of points with
mean density $\rho = \frac{1}{A}$:
$$
L_{A,\Delta}:=2A\Z \cup \left( 2A\Z + \frac{2A\Delta}{1+\Delta}\right).
$$
Furthermore, we denote by $a(A,\Delta)$ and $b(A,\Delta)$ the alternating distances between the points of $L_{A,\Delta}$ (also called periods), i.e.,
\begin{equation} \label{periods}
a(A,\Delta) := \frac{2A\Delta}{1+\Delta} ,
\qquad b(A,\Delta) := \frac{2A}{1+\Delta} . 
\end{equation}
\end{defi}

\begin{remark}[\textbf{Symmetries of the bipartite chain}]
The symmetry transformation $\Delta\to 1/\Delta$ means an exchange
of the periods $a\leftrightarrow b$.
The equidistant case corresponds to $\Delta=1$, the limits $\Delta\to 0^+$
and $\Delta\to\infty$ correspond to the particle coalescence.
For $\Delta\ne 1$, the unit cell contains two particles and the period
of the particle array is $a+b=2A$.
Due to the model's symmetry $\Delta\to 1/\Delta$, $\Delta$ can be restricted
to either $(0,1]$ or $[1,\infty)$ intervals, without any loss of generality.
\end{remark}

The energy per particle for the bipartite lattice can be derived in analogy
with Ref. \cite{Betermin23b}.

\begin{prop}[\textbf{Energy of a bipartite chain}]
Let $f\in \mathcal{F}$, $A>0$ and $\Delta>0$.
Then the $f$-energy of $L_{A,\Delta}$ is $E^{\rm bip}(A,\Delta):=E(L_{A,\Delta})$
given by
\begin{eqnarray} 
E^{\rm bip}(A,\Delta) & = & \frac{1}{2} \int_0^{\infty} \left(
\sum_{j=-\infty}^{\infty} e^{-4j^2 A^2 t} - 1 \right) d\mu_f(t) \nonumber \\
& & + \frac{1}{4} \int_0^{\infty} \sum_{j=-\infty}^{\infty}
\left[ e^{-(a+2jA)^2 t} + e^{-(b+2jA)^2 t} \right] d\mu_f(t) , \label{biprepr2}
\end{eqnarray}
where the spacings $a$ and $b$ of the bipartite chain are given as
the functions of the parameters $A$ and $\Delta$ by {\rm (\ref{periods})}. 
\end{prop}
\begin{remark}[\textbf{Theta functions and energy of an equidistant configuration}]
The sums under integration can be formally expressed in terms of
the Jacobi theta function with nonzero argument \cite{Gradshteyn}
\begin{equation}
\theta_3(u,q) = \sum_{j=-\infty}^{\infty} q^{j^2} e^{2{\rm i}ju} .
\end{equation}

For the equidistant case $\Delta=1$ with $a=b=A$, the above formula
(\ref{biprepr2}) takes the form
\begin{equation} \label{eqformula} 
E^{\rm eq}(A,1) = \frac{1}{2} \int_0^{\infty} \left[
\theta_3\left( e^{-4 A^2 t}\right) - 1 + \theta_2\left( e^{-4 A^2 t}\right)
\right] d\mu_f(t) ,
\end{equation}
where
\begin{equation}
\theta_2(q) = \sum_{j=-\infty}^{\infty} q^{\left(j+{\scriptstyle\frac{1}{2}}\right)^2}
\end{equation}  
is another Jacobi theta function with zero argument \cite{Gradshteyn}. 
Using the relation
\begin{equation}
\theta_2(q) + \theta_3(q) = \theta_3\left(q^{\scriptstyle\frac{1}{4}} \right)
\end{equation}
with $q=\exp(-4A^2 t)$, the formula (\ref{eqformula})
reduces to the previous one (\ref{eqenergy}). 
\end{remark}
\begin{proof}
Let us choose one of the particles as the reference one.
First one sums the interactions with particles to the right,
say at distances $a$, $a+b=2A$, $2a+b=a+2A$, $2a+2b=4A$, etc.: 
\begin{equation}
E^{\rm bip}_R(A,\Delta) = \frac{1}{2} \left[
\sum_{j=1}^{\infty} f(2jA) + \sum_{j=0}^{\infty} f(a+2jA) \right] .  
\end{equation}
The same computation for the interactions with particles to the left at distances $b$, $a+b=2A$, $a+2b=b+2A$, $2a+2b=4A$, etc., gives 
\begin{equation}
E^{\rm bip}_L(A,\Delta) = \frac{1}{2} \left[
\sum_{j=1}^{\infty} f(2jA) + \sum_{j=0}^{\infty} f(b+2jA) \right] .  
\end{equation}
The total energy per particle $E^{\rm bip}(A,\Delta)$ is the sum
of the right and left contributions,
$E^{\rm bip}(A,\Delta) = E^{\rm bip}_R(A,\Delta) + E^{\rm bip}_L(A,\Delta)$, i.e.,
\begin{equation} \label{energybip}
E^{\rm bip}(A,\Delta) = \frac{1}{2} \left[
2 \sum_{j=1}^{\infty} f(2jA) + \sum_{j=0}^{\infty} f(a+2jA) + 
\sum_{j=0}^{\infty} f(b+2jA) \right] .  
\end{equation}

The sums on the right-hand side of the above expression are restricted
to positive integer values (including 0) of $j$ only, which is impractical
in view of general algebraic manipulations.
However, they can be extended to all (positive and negative) integers. 
The first sum can be written as
\begin{equation}
2 \sum_{j=1}^{\infty} f(2jA) = \xsum_{j=-\infty}^{\infty} f(2jA) . 
\end{equation}
To rewrite second and third sums, let us consider the following sum
\begin{equation}
\sum_{j=-\infty}^{\infty} f(a+2jA) = \sum_{j=0}^{\infty} f(a+2jA) +
\sum_{j=-\infty}^{-1} f(a+2jA) .  
\end{equation}
Since
\begin{eqnarray}
\sum_{j=-\infty}^{-1} f(a+2jA) & = & \sum_{j=1}^{\infty} f(a-2jA)
= \sum_{j=1}^{\infty} f(a-2A-2(j-1)A) \nonumber \\
& = & \sum_{j=0}^{\infty} f(-b-2jA) = \sum_{j=0}^{\infty} f(b+2jA) ,
\end{eqnarray}
it holds that
\begin{equation}
\sum_{j=0}^{\infty} f(a+2jA) + \sum_{j=0}^{\infty} f(b+2jA) 
= \sum_{j=-\infty}^{\infty} f(a+2jA) .
\end{equation}
One can prove in a similar way that
\begin{equation}
\sum_{j=0}^{\infty} f(a+2jA) + \sum_{j=0}^{\infty} f(b+2jA) 
= \sum_{j=-\infty}^{\infty} f(b+2jA) .
\end{equation}
We conclude that the energy per particle of the bipartite array
(\ref{energybip}) can be written in a symmetrized form as follows
\begin{equation} \label{biprepr1}
E^{\rm bip}(A,\Delta) = \frac{1}{2} \xsum_{j=-\infty}^{\infty} f(2jA)
+ \frac{1}{4} \sum_{j=-\infty}^{\infty} \left[ f(a+2jA) + f(b+2jA) \right] .  
\end{equation}

Finally, inserting the integral representation of the pair potential
(\ref{intpot}) into the formula for the bipartite energy (\ref{biprepr1}),
we arrive at the desired formula.
\end{proof}

\subsection{Transition between the equidistant and bipartite ground states}
\label{Sec23}
In statistical mechanics, interacting many-body systems in thermal
equilibrium can undertake, as a result of the change of external conditions
such as temperature or external fields, a transition from one state of
a medium to another one at some special transition point.
Ehrenfest's classification of phase transitions is based on the behavior
of the free energy as a function of thermodynamic variables, namely by
the lowest derivative of the free energy that is discontinuous at
the transition point \cite{Jaeger98}.
Solid/liquid/gas transitions in interacting particle systems are usually
of first order because of a discontinuous change in particle density
which is given by a first derivative of the free energy.
Second-order phase transitions in simple magnetic spin materials, which occur
at a critical temperature $T_c$, are continuous in the first derivative of
the free energy per particle $f(T,H)$ with respect to the external magnetic
field $H$, i.e. the magnetization $M\propto \partial f(T,H)/\partial H$,
but discontinuous in the second derivative of $f(T,H)$ with respect to
temperature $T$ or magnetic field $H$.

The free energy is usually defined as a function of $T$ and $M$ that we will note $f(T,M)$.
Respecting the symmetry $f(T,M)=f(T,-M)$, the expansion of the free energy
in powers of small magnetization $M\to 0$ reads as, for instance
up to fourth order,
\begin{equation} \label{Landau1}
f(T,M) = f_0(T) + f_2(T) M^2 + f_4(T) M^4 + O(M^6) ,
\end{equation}
where the coefficients $f_0(T)$, $f_2(T)$, $f_4(T)$, etc. are complicated
continuous functions of temperature $T$ which can be usually obtained only
approximately, e.g. by using Landau's mean-field theory
\cite{Bausch72,Landau37}.
The actual value of $M$ is determined by the minimization condition
for the free energy:
\begin{equation} \label{Landaumin1}
\frac{\partial}{\partial M} f(T,M) = 0 , \qquad
\frac{\partial^2}{\partial M^2} f(T,M) > 0 .
\end{equation}
Concretely, the first equation gives, as $M\to 0$,
\begin{equation} \label{min1}
2 f_2(T) M + 4 f_4(T) M^3 + O(M^5) = 0 .
\end{equation}
General experience tells us that the coefficient
$f_2(T)$ can have both $\pm$ signs:  
\begin{equation}
f_2(T) > 0 \quad \mbox{for $T>T_c$} \qquad {\rm and}
\quad f_2(T) < 0 \quad \mbox{for $T<T_c$} ,
\end{equation}
where $T_c$ is a ``critical'' temperature and $f_4(T)>0$ in the neighbourhood
of $T=T_c$.
If $T>T_c$, both $f_2(T)$ and $f_4(T)$ are positive and equation (\ref{min1})
has the only trivial solution $M=0$ corresponding to the minimum of
the free energy of the ``disordered'' phase.
If $T<T_c$, the inequalities $f_2(T)<0$ and $f_4(T)>0$ cause that
equation (\ref{min1}) has three real solutions $M\in \{0,\pm M_0\}$ where
\begin{equation}
M_0 = \sqrt{\frac{-f_2(T)}{2 f_4(T)}} .
\end{equation}  
Now the trivial solution $M=0$ corresponds to the maximum of the free energy
while the two ``symmetry broken'' solutions $M=\pm M_0$ correspond
to the minimum of the free energy.

An analogous classification scheme of phase transitions can be made for
our 1D system of interacting particles in the ground state.
The role of the free energy as a function of the temperature $T$ and
the magnetization $M$ is played by the bipartite ground-state energy
$E^{\rm bip}(A,\Delta)$ as the function of the inverse mean density $A$ and
the parameter $\Delta$ of the bipartite lattice.

The above bipartite formula for the energy (\ref{biprepr2}) can be used
to derive exact formulas for a second-order phase transition between
the equidistant and bipartite ground states, with a continuous change of
the parameter $\Delta$ from its equidistant value 1 to
a bipartite one $\Delta\ne 1$.
To account for the invariance of bipartite energy (\ref{biprepr2})
with respect to the transformation $\Delta\to 1/\Delta$, we introduce
the new parameter $\varepsilon$ via
\begin{equation}
\Delta = e^{\varepsilon} .
\end{equation}
This change of variables $\Delta\mapsto \varepsilon$ maps the interval $(0,1)$ onto negative values $(-\infty,0)$ and the interval $(1,+\infty)$ onto positive values $(0,\infty)$.
The equidistant configuration with the parameter $\Delta=1$ corresponds to $\varepsilon=0$
and the bipartite configuration with $\Delta\ne 1$ corresponds to
$\varepsilon\ne 0$.
The periods of the bipartite chain (\ref{periods}) are now given by
\begin{equation} \label{periodseps}
a(A,\varepsilon) = \frac{2A}{1+e^{-\varepsilon}} ,
\qquad b(A,\varepsilon) = \frac{2A}{1+e^{\varepsilon}}   
\end{equation}
and the bipartite energy exhibits the $\varepsilon\to -\varepsilon$ symmetry
which corresponds to an exchange of the periods $a\leftrightarrow b$.
The analytic expansion of $E^{\rm bip}(A,e^{\varepsilon})$ around the equidistant
value $\varepsilon=0$, in powers of small $\varepsilon$, thus contains
only even powers of $\varepsilon$. More precisely, we find the following.

\begin{prop}[\textbf{General expansion of the energy of a bipartite chain}]\label{prop:expansion}
For any $A>0$, we have, as $\varepsilon\to 0$, 
\begin{equation} \label{Landau2}
E^{\rm bip}(A,e^{\varepsilon}) = E^{\rm eq}(A) + E_2(A) \varepsilon^2
+ E_4(A) \varepsilon^4 + E_6(A) \varepsilon^6 + O\left( \varepsilon^8 \right) ,
\end{equation}
where the equidistant energy $E^{\rm eq}(A)$ is given by {\rm (\ref{eqformula})}
and  
\begin{eqnarray}
E_2(A) & = & - \frac{A^2}{4} \int_0^{\infty} \left[
\frac{t}{2} \theta_2\left( e^{-4 A^2 t} \right)
+ t^2 \theta_2^{(1)}\left( e^{-4 A^2 t} \right) \right] d\mu_f(t) , \nonumber\\
E_4(A) & = & \frac{A^2}{16} \int_0^{\infty} \Bigg\{
\left( \frac{t}{3} + \frac{A^2 t^2}{4} \right) \theta_2\left( e^{-4 A^2 t} \right)
\nonumber \\ & &
+ \left( \frac{2 t^2}{3} + A^2 t^3 \right) \theta_2^{(1)}\left( e^{-4 A^2 t} \right)
+ \frac{1}{3} A^2 t^4 \theta_2^{(2)}\left( e^{-4 A^2 t} \right)
\Bigg\} d\mu_f(t) , \nonumber \\
E_6(A) & = & -\frac{A^2}{4} \int_0^\infty \Biggl\{\left(\frac{17}{1440}t
+\frac{1}{48}A^2t^2 +\frac{1}{192}A^4t^3\right)\theta_2\left(e^{-4A^2t}\right)
\nonumber\\ & & + \left(\frac{17t^2}{720}+\frac{A^2t^3}{12}
+\frac{A^4t^4}{32}\right)\theta_2^{(1)}\left(e^{-4A^2t}\right) \nonumber\\
& & + \left(\frac{A^2t^4}{36} + \frac{A^4t^5}{48}\right)
\theta_2^{(2)}\left(e^{-4A^2t}\right) +\frac{A^4t^6}{360}
\theta_2^{(3)}\left(e^{-4A^2t}\right) \Biggr\} {\rm d}\mu_f(t) . \nonumber \\ & &
\end{eqnarray}
Here, for all $ n\in \N$, we define
\begin{eqnarray}
\theta_2^{(n)}\left( e^{-4 A^2t}\right) & := &
\frac{\partial^n}{\partial t^n} \theta_2\left( e^{-4 A^2 t} \right) \nonumber \\
& = & (-1)^n 4^n A^{2n} \sum_{j=-\infty}^{\infty}
\left( j+\frac{1}{2} \right)^{2n} e^{-4\left(j+\scriptstyle\frac{1}{2}\right)^2A^2 t} . \label{dertheta2}
\end{eqnarray}
\end{prop}
\begin{proof}
While the first integral on the right-hand side of (\ref{biprepr2}) does not
depend on $\varepsilon$, the second one does.
Let us consider the small-$\varepsilon$ expansion of the function
\begin{eqnarray}
\hspace{-17pt}
e^{-(a+2jA)^2 t} + e^{-(b+2jA)^2 t} = e^{-4\left(j+\scriptstyle\frac{1}{2}\right)^2A^2 t} 
\Bigg\{ 2 + \left[ - \frac{A^2 t}{2} + 4 (A^2 t)^2
  \left( j + \frac{1}{2} \right)^2 \right] \varepsilon^2 \nonumber \\
\hspace{-17pt}
+ A^2 t \left[ \frac{1}{12} + \frac{1}{16} A^2 t
- A^2 t \left( \frac{2}{3} + A^2 t \right) \left( j + \frac{1}{2} \right)^2 
+ \frac{4}{3} (A^2 t)^3 \left( j + \frac{1}{2} \right)^4 \right]
\varepsilon^4 \nonumber \\ \hspace{-17pt}
+ A^2t \Biggl[ -\frac{17}{1440} - \frac{1}{48}A^2t - \frac{1}{192}(A^2t)^2
+ A^2 t\left( \frac{17}{180}+\frac{A^2t}{3}+\frac{(A^2t)^2}{8}\right)
\left( j+\frac{1}{2}\right)^2 \nonumber\\ \hspace{-17pt}
- (A^2t)^3 \left(\frac{4}{9} + \frac{A^2t}{3} \right)
\left( j+\frac{1}{2}\right)^4
+\frac{8(A^2t)^5}{45}\left( j+\frac{1}{2}\right)^6\Biggr]\varepsilon^6
+ O(\varepsilon^8) \Bigg\} .
\end{eqnarray}
Using the notation (\ref{dertheta2}) for the derivatives of the
$\theta_2$-function, this gives the expansion of the bipartite energy
in even powers of $\varepsilon$ as desired.
\end{proof}

The general analysis of the series expansions (\ref{Landau1}) together
with the minimization condition (\ref{Landaumin1}) or (\ref{min1})
in the context of critical phenomena in equilibrium statistical mechanics
is explained in many textbooks, see e.g. \cite{Landau80,Toledano87}.

By (\ref{Landau2}), it turns out that, for all $A>0$, as $\varepsilon\to 0$,
we have
$$
E^{\rm bip}(A,e^{\varepsilon})-E^{\rm eq}(A)=\varepsilon^2\left[E_2(A)
+E_4(A)\varepsilon^2 + O(\varepsilon^4) \right],
$$
which means that the sign of $E_2(A)$ determines the one of
$E^{\rm bip}(A,e^{\varepsilon})-E^{\rm eq}(A)$ for sufficiently small values
of $\varepsilon>0$, giving the optimality of the equidistant configuration
$\{jA\}_{j\in \Z}$ when $E^{\rm bip}(A,e^{\varepsilon})-E^{\rm eq}(A)>0$ and
the optimality of the (strictly) bipartite configuration when
$E^{\rm bip}(A,e^{\varepsilon})-E^{\rm eq}(A)<0$.

In the following, we recall the definitions of critical
(second-order transition) and tricritical points as well as sufficient
conditions for their existence, as given in \cite{Betermin24}.
The proofs we gave in the previous work are general enough
to show the same results in the present setting.

\begin{defi}[\textbf{Transition and tricritical points}]
Let $f\in \mathcal{F}$, then any $A_c$ such that
$$
E_2(A_c)=0
$$
with a change of sign for $E_2$ at $A=A_c$ is called a transition point.
Furthermore, if $f=f_\alpha$ depends on a real parameter $\alpha$,
we say that $A^t>0$ and $\alpha^t$ are coordinates of a tricritical point
if $A^t(\alpha^t)$ is a transition point satisfying
$$
E_2(A^t(\alpha^t))=E_4(A^t(\alpha^t))=0,
$$
with a change of sign for $E_4$ at $A^t(\alpha^t)$.
\end{defi}

\begin{prop}[\textbf{Minimal configurations and asymptotics near a transition point, \cite{Betermin24}}] \label{prop:exp}
Let $f\in \mathcal{F}$ and $A_c$ be a transition point such that:
\begin{enumerate}
\item[(1)] $E_2$ is strictly decreasing in the neighborhood of $A_c$;
\item[(2)] $E_4(A_c)>0$.
\end{enumerate}
Then there exists $A_0>0$ and $A_1>0$ such that
\begin{itemize}
\item[•] if $A_0<A<A_c$, then $\varepsilon=0$ is the unique minimizer of
$\varepsilon\mapsto E^{\rm bip}(A,e^{\varepsilon})$;
\item[•] if $A_c<A<A_1$, a minimizer $\varepsilon$ of
$\varepsilon\mapsto E^{\rm bip}(A,e^{\varepsilon})$ satisfies the following
asymptotics, as $A\to A_c$:
\begin{equation}\label{nontrivialeps}
\varepsilon=\sqrt{\frac{-\frac{d E_2}{d A}(A_c)}{2E_4(A_c)}}\sqrt{A-A_c}
+ o(\sqrt{A-A_c}).
\end{equation}
\end{itemize}
\end{prop}

\begin{remark}
In the theory of critical phenomena \cite{Baxter07,Samaj13}, the
order parameter is a singular function of the deviation from the critical
point, in our case $\varepsilon$ is of order $(A-A_c)^{\beta}$ with $\beta=\frac{1}{2}$
being the mean-field critical exponent.
In the ordered region, the difference between the equidistant and
bipartite energies is proportional to
$(A-A_c)^{2-\alpha}$ with the mean-field critical exponent $\alpha=0$.
As expected for a second-order phase transition, we note that the first (resp. the second) derivation of the energy with respect to $A$ is continuous (resp. discontinuous) at $A=A_c$.
\end{remark}

Let us assume that the interaction potential $f=f_{\alpha}$ depends
on a parameter $\alpha\in\mathbb{R}$.
Then, in an interval of $\alpha$-values, say for $\alpha>\alpha^t$
with $\alpha^t$ being an edge value of $\alpha$,
to each $\alpha$ there exists a critical point $A_c(\alpha)$, given by
the nullity of the coefficient $E_2\left( A_c(\alpha)\right)=0$ in the
series expansion (\ref{Landau2}).
In analogy with statistical mechanics \cite{Bausch72,Griffiths70,Landau80},
only one continuous curve of second-order phase transitions is expected.
This curve terminates at the tricritical point characterized by
the edge value of $\alpha=\alpha^t$ and the critical value of
$A=A^t(\alpha^t)$.
According to the Landau's theory of phase transitions \cite{Landau37,Landau80},
the tricritical point is determined by the nullity of the first
two coefficients of the expansion (\ref{Landau2}).
More precisely, we have the following result.

\begin{prop}[\textbf{Asymptotics near a tricritical point, \cite{Betermin24}}]
Let $f_\alpha\in \mathcal{F}$ depending on a real parameter $\alpha$ and
$(\alpha^t,A^t)$ be the coordinates of a tricritical point such that
\begin{enumerate}
\item[(1)] $E_2$ is strictly decreasing in the neighborhood of $A^t$;
\item[(2)] $E_6(A^t)=\frac{d^6}{d\varepsilon^6}
[E^{\rm bip}(A^t,e^{\varepsilon})]_{|\varepsilon=0}>0$.
\end{enumerate}
Then there exists $A_2>0$ such that for $A^t<A<A_2$, a minimizer of
$\varepsilon\mapsto E^{\rm bip}(A,e^{\varepsilon})$ satisfies the following
asymptotics as $A\to A^t$:
$$
\varepsilon=\sqrt[4]{\frac{-\frac{d E_2}{d A}(A^t)}{3E_6(A^t)}}
(A-A^t)^{\frac{1}{4}}+o((A-A^t)^{\frac{1}{4}}).
$$
\end{prop}

\begin{remark}[\textbf{On first-order phase transitions}]
The phase transitions are of first order when $\alpha<\alpha^t$, i.e.
the equidistant value of the parameter $\varepsilon=0$ skips
discontinuously to some $\varepsilon>0$ at the transition point.
There is no explicit formula for a first-order transition and the curve of
first-order phase transitions can be obtained only numerically by comparing
the equidistant energy with the bipartite energy in a larger interval
of the $\varepsilon$ values.
\end{remark}

\renewcommand{\theequation}{3.\arabic{equation}}
\setcounter{equation}{0}

\section{One-dimensional Lennard-Jones model} \label{Sec3}
This section is devoted to the application of the general theory of phase
transitions to the 1D Lennard-Jones model of particles
interacting via the Mie potential $f_{nm}$ given by (\ref{Mie}).

\subsection{Interaction energies and series expansions}
Using the tools from section \ref{Sec2}, we get the following result
concerning the $f_{mn}$-energy.
\begin{prop}[\textbf{Energy of a bipartite chain for the Mie potential}] \label{prop:enerbipLJ}
Let $f=f_{nm}$ be defined by {\rm (\ref{Mie})}, where $n>m$, then,
for all $A>0$, the energy per particle for the bipartite configuration is
\begin{equation} \label{energyy}
E_{nm}^{\rm bip}(A,\Delta) = \frac{1}{n-m}
\left[ m\ U(n,A,\Delta)-n\ U(m,A,\Delta) \right],
\end{equation}
where the energy per particle {\rm (\ref{biprepr1})} for
the Riesz interaction potential $f:r\mapsto r^{-s}$ reads as
\begin{equation} \label{Riesz}
U(s,A,\Delta) = \frac{1}{(2A)^s} \zeta(s) + \frac{1}{2 (2A)^s}
\left[ \zeta\left( s,\frac{\Delta}{1+\Delta}\right)
+\zeta\left( s,\frac{1}{1+\Delta}\right) \right]
\end{equation}
with 
\begin{equation} \label{RiemannZeta}
\zeta(s) = \sum_{j=1}^{\infty} \frac{1}{j^s} , \qquad \Re(s)>1
\end{equation}
denoting the standard Riemann zeta function and
\begin{equation} \label{HurwitzZeta}
\zeta(s,a) = \sum_{j=0}^{\infty} \frac{1}{(j+a)^s} , \qquad \Re(s)>1,
\qquad a>0,  
\end{equation}
is the Hurwitz zeta function. \\
In particular, the energy per particle for the equidistant configuration
$\{j A\}_{j\in \Z}$ is given by
\begin{equation} \label{energd1}
E_{nm}^{\rm eq}(A) = \frac{1}{n-m} \left[ \frac{m\, \zeta(n)}{A^n}
-\frac{n\ \zeta(m)}{A^m} \right].
\end{equation}
Furthermore $A\mapsto E_{nm}^{\rm eq}(A)$ is decreasing on $(0,A_{\min}^{nm})$
and increasing on $(A_{\min}^{nm},+\infty)$ where
\begin{equation} \label{A}
A^{nm}_{\min} := \left[ \frac{\zeta(n)}{\zeta(m)} \right]^{\frac{1}{n-m}} ,
\qquad E_{nm}^{\rm eq}(A^{nm}_{\min}) = - \zeta(n)^{\frac{m}{m-n}} \zeta(m)^{\frac{n}{n-m}} .
\end{equation}
Moreover, we have
$$
\lim_{n\to +\infty} A^{nm}_{\min}=1 \quad \textnormal{and} \quad 
A^m_{\min} :=\lim_{n\to m^+}  A^{nm}_{\min}  = \exp\left[ \frac{\zeta'(m)}{\zeta(m)} \right],  
$$
and
$$
\lim_{m\to +\infty} A^m_{\min}=1.
$$
\end{prop}
\begin{proof}
The only points which have to be proven are the two last ones.
We have indeed
$$
\lim_{n\to +\infty} A^{nm}_{\min}=\lim_{n\to +\infty}
\left[ \frac{\zeta(n)}{\zeta(m)} \right]^{\frac{1}{n-m}}=1,
$$
since $\displaystyle \lim_{n\to +\infty} \zeta(n)=1$, and
$$
\lim_{n\to m^+} A^{nm}_{\min} = \lim_{\varepsilon\to 0^+}
\left[ \frac{\zeta(m+\varepsilon)}{\zeta(m)} \right]^{1/\varepsilon}
=\exp\left[ \frac{\zeta'(m)}{\zeta(m)} \right].
$$
Furthermore, we get
$$
\lim_{m\to +\infty} \frac{\zeta'(m)}{\zeta(m)}= \lim_{m\to +\infty}
\frac{-\displaystyle\sum_{k\geq 2} \ln(k) k^{-m}}{
1+\displaystyle\sum_{k\geq 2 }k^{-m}}=0,
$$
which gives the last asymptotics.
\end{proof}
\begin{remark}[\textbf{Analytic continuation of the energy}]\label{rmk_analytic}
For more information about the Riemann and Hurwitz zeta functions,
see for instance \cite{Gradshteyn}.
These special functions can be analytically continued to the whole complex
$s$-plane, except for the simple pole at $s=1$.
This permits one to extend the definition of the above energies per particle
to all values of the LJ parameters $n>m$ such that $n,m\ne 1$.
Physically acceptable models must respect the inequality $n>m$.
For the case $0<m<1<n$, the possibility of crystallization for a large but
finite number of particles was analyzed recently in \cite{Crismale23}.
Even more exotic case of LJ numbers $-1<m<n$ was studied in \cite{Luo21}.
\end{remark}
\begin{remark}[\textbf{Numerical study of the behavior of $A^m_{\min}$}]
The plot of $A^m_{\min}$ versus $m$ is pictured in figure \ref{f2}.
It is seen that $A^m_{\min}$ goes to 0 in the limit $m\to 1^+$ which
follows from the asymptotic relation
\begin{equation}
\frac{\zeta'(m)}{\zeta(m)} = - \frac{1}{m-1} + O(1) . 
\end{equation}
$A^m_{\min}$ goes to 1 when $m\to\infty$ which is rigorously proven in
the previous result.
\end{remark}

\begin{figure}[tbp]
\begin{center}
\includegraphics[clip,width=0.5\textwidth]{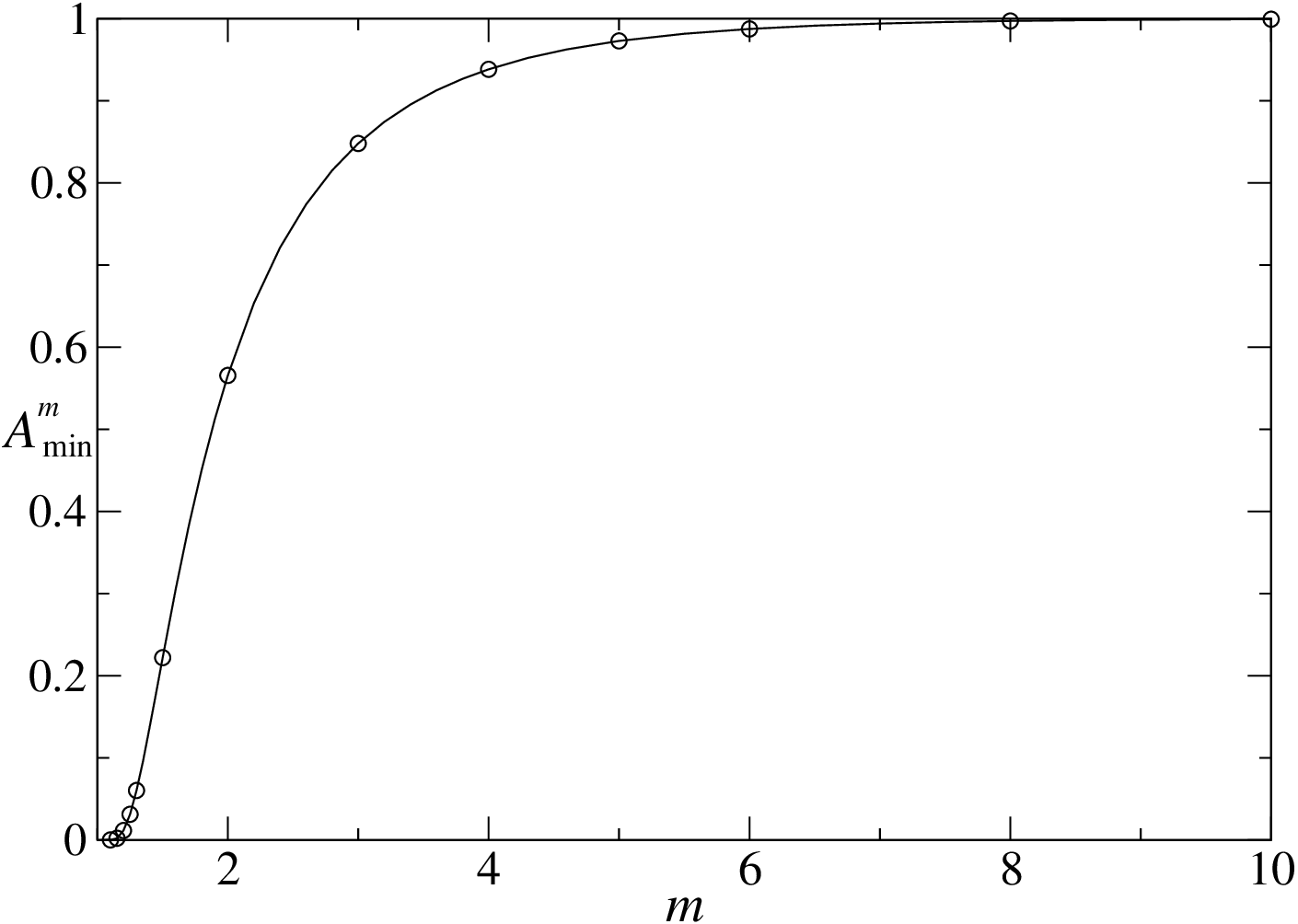}
\caption{The 1D $(n,m)$ LJ model in the limit $n\to m^+$.
The plot of the lattice spacing $A^m_{\min}$, at which the equidistant energy
per particle has the global minimum, versus $m$.
$A^m_{\min}$ goes to 0 in the limit $m\to 1^+$ and it is numerically confirmed
that the limit is 1 when $m\to\infty$.}
\label{f2}
\end{center}
\end{figure}

\begin{example}[\textbf{The standard Lennard-Jones potential}]
For the standard $(12,6)$ Lennard-Jones model, one gets
\begin{eqnarray} 
A^{12,6}_{\min} = \frac{\pi}{5^{\frac{1}{3}}\sqrt{3} } \left( \frac{691}{1001} \right)^{\frac{1}{6}}
\approx 0.997179263885 , \nonumber \\
E_{12,6}^{\rm eq}(A^{12,6}_{\min}) = - \frac{715}{691} \approx - 1.034732272069 . 
\end{eqnarray}
These results agree with those obtained numerically
 in \cite{Stillinger95}; recall that Stillinger used the potential
$f(r)=4(r^{-12}-r^{-6})$ implying that $r_{\min}=2^{1/6}$ and $f(r_{\min})=-1$,
so our length has to be scaled as follows $2^{1/6}A^{12,6}_{\min}$
to reproduce his value approximately equal to $1.1193$.
The value of the spacing $A^{12,6}_{\min}$ is very close to 1
which is the distance at which the LJ potential exhibits its minimum $-1$.
\end{example}

Applying directly the expansion constructed in Proposition
\ref{prop:expansion}, we have the following finding.
\begin{corollary}[\textbf{Asymptotic expansion of the energy and the transition point}]
Let $A>0$, $m>n>1$, then, as $\varepsilon\to 0$, we have
\begin{equation} \label{trans}
E_{nm}^{\rm bip}\left( A,e^{\varepsilon}\right) = E_{nm}^{\rm eq}(A)
+ E_2^{nm}(A) \varepsilon^2 + E_4^{nm}(A) \varepsilon^4 + O(\varepsilon^6), 
\end{equation}
where the absolute term is given by (\ref{energd1}) and the coefficients
read as
\begin{eqnarray}
E_2^{nm}(A) & = & \frac{n m}{n-m} \frac{1}{2^5} \left[
\frac{2^{2+n}-1}{(2A)^n} (1+n) \zeta(n+2) \right. \nonumber \\
& & \left. - \frac{2^{2+m}-1}{(2A)^m} (1+m) \zeta(m+2) \right] , \label{rov1}
\end{eqnarray}
\begin{eqnarray}
E_4^{nm}(A) & = & \frac{n m}{3(n-m)} \frac{1}{2^{11}} \left[
\frac{2^{4+n}-1}{(2A)^n} (1+n) (2+n) (3+n) \zeta(n+4) \right. \nonumber \\
& & \left.  - \frac{2^{4+m}-1}{(2A)^m} (1+m) (2+m) (3+m) \zeta(m+4) \right]
\nonumber \\ & & - \frac{1}{6} E_2^{nm}(A). \label{rov2}
\end{eqnarray}  
Furthermore, the unique transition point $A_c^{nm}$ is given by
\begin{equation} \label{astar}
A_c^{nm} = \frac{1}{2} 
\left[ \frac{(2^{2+n}-1)(1+n)\zeta(n+2)}{(2^{2+m}-1)(1+m)\zeta(m+2)}
\right]^{\frac{1}{n-m}} .  
\end{equation} 
Moreover, in the limit $n\to m^+$, this relation yields
\begin{equation}
A_c^{m^+,m}:=\lim_{n\to m^+} A_c^{nm}=\exp\left( \frac{\ln 2}{2^{2+m}-1}+\frac{1}{1+m}
+\frac{\zeta'(m+2)}{\zeta(m+2)} \right),
\end{equation}
and in the hard-core limit $n\to\infty$, the critical inverse particle density
tends to the $m$-independent value
\begin{equation}
\lim_{n\to\infty} A_c^{nm} = 1.
\end{equation}  
\end{corollary}
\begin{proof}
These are direct computations from the formulas given in Proposition
\ref{prop:expansion}.
The value of $A_c^{nm}$, namely the solution of $E_2^{nm}(A_c^{nm})=0$,
is easily computed from (\ref{rov1}).
Furthermore, $E_2^{nm}$ indeed exhibits a change of sign in the neighborhood
of $A_c^{nm}$ by a direct computation.
Moreover, we have, defining $f:t\mapsto (2^{2+t}-1)(1+t)\zeta(t+2)$,
which is an analytic function on $(1,+\infty)$,
$$
\lim_{n\to m^+} \frac{1}{2}\left[ \frac{f(n)}{f(m)} \right]^{\frac{1}{n-m}}
= \lim_{h\to 0^+} \frac{1}{2}\exp\left(\frac{1}{h}\ln\left(\frac{f(m+h)}{f(m)}
\right) \right) = \frac{1}{2}\exp\left(\frac{f'(m)}{f(m)} \right)
$$
which gives our result after straightforward computation.
We also easily get the last asymptotics as we get the one of $A_{\min}^{nm}$
as $n\to +\infty$ in Proposition \ref{prop:enerbipLJ}.
\end{proof}
\begin{example}[\textbf{Numerics of the optimal equidistant energy and the transitions points}]
The plots of the ground-state energy versus the mean lattice spacing $A$
for $m=6$ and two values of $n\in\{12,7\}$ are represented in figure \ref{f1}.
The ground state is equidistant for all $A\in(0,A_c]$ where the critical
values of $A_c$ indicate the second-order phase
transitions from the equidistant to bipartite ground states.
It is seen that the plots $E(A)$ are non-monotonous with just one global
minimum $A_{\min}^{nm}$ represented by the black circle/square for $n\in\{12,7\}$,
respectively, located in the equidistant part of the phase diagram. 
For the $(12,6)$ case one obtains $A_c^{12,6}\approx 1.108654785157924$ and for
$(7,6)$ one has $A_c^{7,6}\approx 1.1427384940215781$.
\end{example}

\begin{figure}[tbp]
\begin{center}
\includegraphics[clip,width=0.5\textwidth]{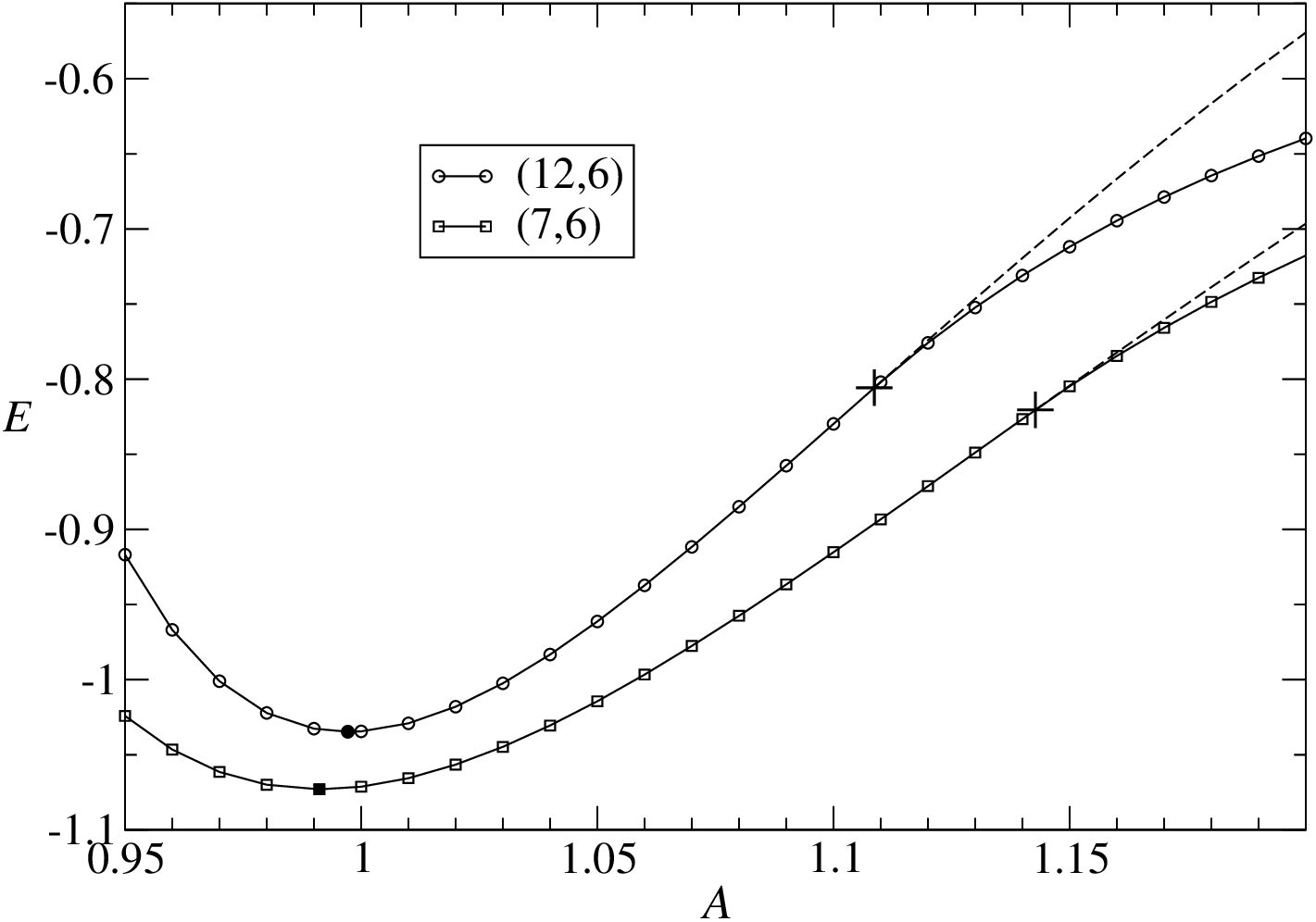}
\caption{The ground-state energy $E$ as the function of the mean
lattice spacing $A$ for the 1D LJ model with $m=6$; the curves for $n=12$
and $n=7$ are represented by white circles and squares, respectively. The critical
values of $A_c$ are marked by crosses.
The black circle and square correspond to the absolute minimum of the energy
for $n=12$ and $n=7$, respectively. The dashed curves correspond to a
prolongation of the equidistant energy to the region $A>A_c^{nm}$
where it is not the ground state energy.}
\label{f1}
\end{center}
\end{figure}

\begin{remark}[\textbf{Numerical test of the mean-field asymptotics}]
The mean-field prediction of the singular behavior of the order
parameter (\ref{nontrivialeps}) close to the critical $A_c^{nm}$
can be tested numerically by using (\ref{adel}).
Since $\Delta\sim 1+\varepsilon$ for very small $\varepsilon$, 
we plot the numerical data for $\Delta-1$ versus $A-A_c^{nm}$ in
the logarithmic scale in figure \ref{f4}.
The linear fits of data are represented by solid lines, the slope $\beta$
equals to approximately $0.500081$ for the $(12,6)$ case and
to approximately $0.500091$ for the $(7,6)$ case, confirming
the mean-field critical exponent $\frac{1}{2}$.
\end{remark}

\begin{figure}[tbp]
\begin{center}
\includegraphics[clip,width=0.5\textwidth]{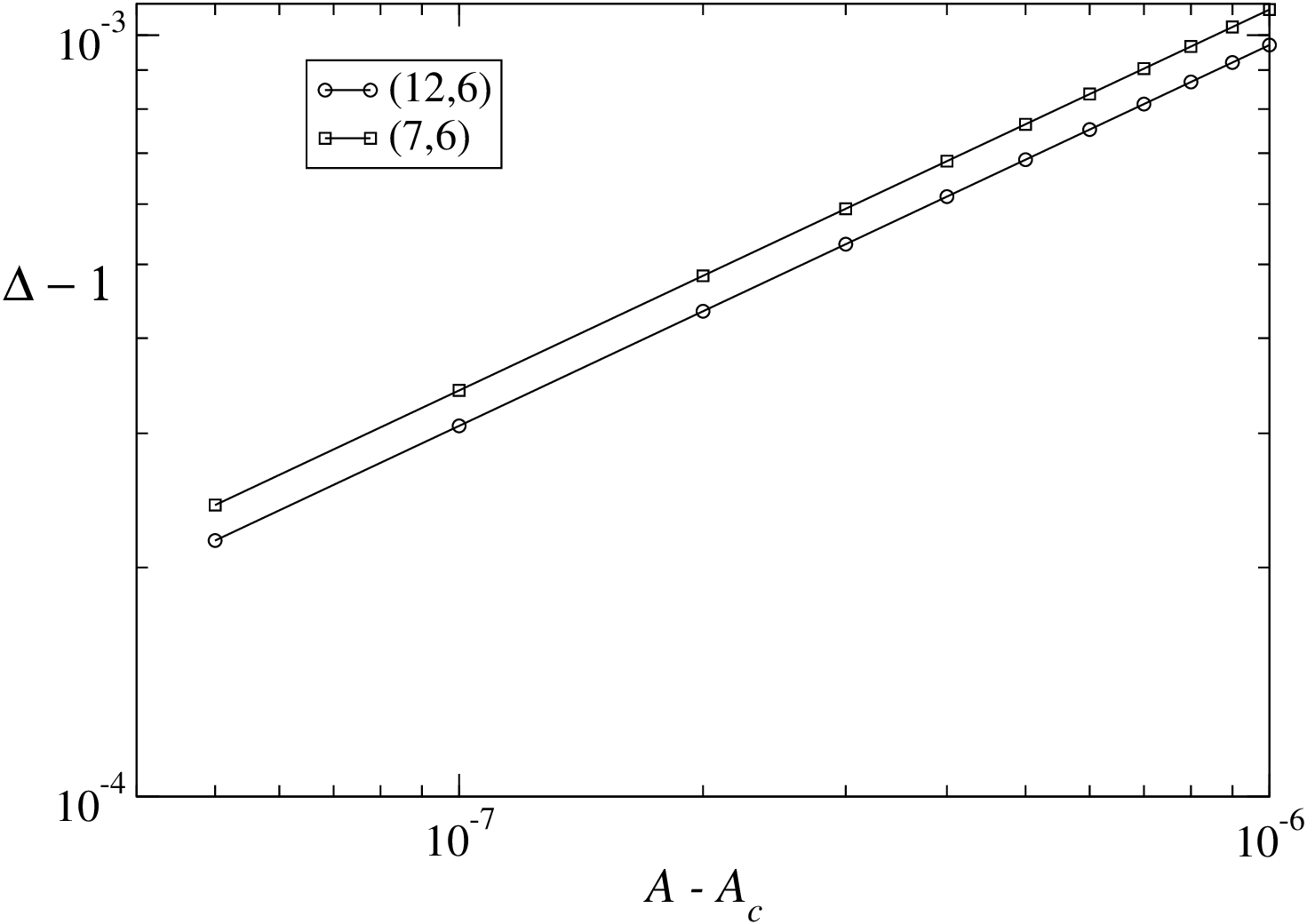}
\caption{Numerical data for $\Delta-1$ versus $A-A^{nm}_c$,
in the critical region $A-A^{nm}_c\ll 1$, for the LJ model with $m=6$,
white circles correspond to $n=12$ and white squares to $n=7$, in the
logarithmic scale.
The solid lines are linear fits of data, the slope $\beta$ approximately
equals to $0.500081$ for $n=12$ and to $0.500091$ for $n=7$.}
\label{f4}
\end{center}
\end{figure}

\subsection{Transition between the equidistant and bipartite chains}
\label{Sec33}
Using (\ref{astar}) we can easily find the simple phase diagram $[n,A]$
with fixed $m=6$, plotted in figure \ref{f3}. 
Here the line $n$ vs. $A_c^{n6}$ marks the border between equidistant
$\Delta=1$ and non-equidistant $\Delta\ne 1$ phases.
In the limit $n\to\infty$ we approach the hard-core limit with unit radius
and $A_c\to 1^+$ as expected.

\begin{figure}[tbp]
\begin{center}
\includegraphics[clip,width=0.5\textwidth]{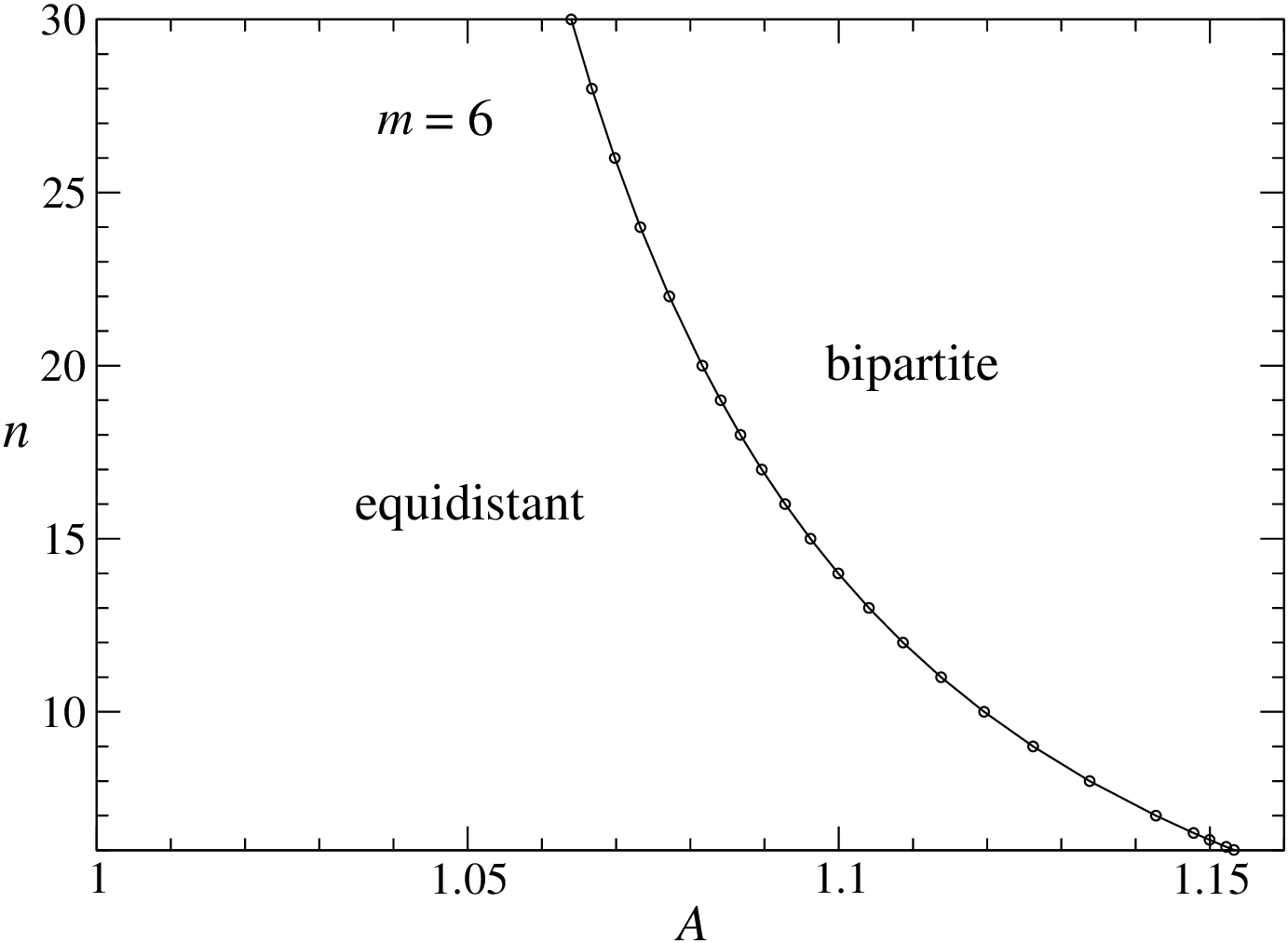}
\caption{The phase diagram of the 1D LJ model with $m=6$, in the $(A,n)$
($n>6$) plane.
The solid curve represents the phase boundary of critical points $A_c^{n6}$
between the equidistant and bipartite ground states.}
\label{f3}
\end{center}
\end{figure}

\begin{prop}[\textbf{Absence of tricritical point}]
Let $n>m$, then there is no tricritical solution $A^t>0$ to the equations
$E_2^{nm}(A^t)=E_4^{nm}(A^t)=0$, i.e. all phase transitions in
the 1D $(n,m)$ LJ model are of second order.
\end{prop}
\begin{proof}
Formula (\ref{astar}) tells us that 
$$
E_2^{nm}(A^t)=0 \iff A^t = \frac{1}{2} 
\left[ \frac{(2^{2+n}-1)(1+n)\zeta(n+2)}{(2^{2+m}-1)(1+m)\zeta(m+2)}
\right]^{\frac{1}{n-m}}.
$$
On the other hand, defining $u(t)=(2^{4+t}-1)(1+t)(2+t)(3+t)\zeta(t+4)$,
we get from (\ref{rov2}) that
$$
E_4^{nm}(A^t) = \frac{C}{(2A^t)^m} \left[\frac{u(n)}{(2A_0)^{n-m}}-u(m) \right],
\quad C:=\frac{nm}{3\times 2^{11+m}(n-m)}.
$$
Let us assume that $E_4^{nm}(A^t)=0$, then we would have
$$
A^t=\frac{1}{2}\left[ \frac{u(n)}{u(m)} \right]^{\frac{1}{n-m}} .
$$
The two above formulas for $A^t$ are equivalent when
$$
\frac{(2^{2+n}-1)(1+n)\zeta(n+2)}{(2^{2+m}-1)(1+m)\zeta(m+2)}
=\frac{(2^{4+n}-1)(1+n)(2+n)(3+n)\zeta(n+4)}{(2^{4+m}-1)(1+m)(2+m)(3+m)\zeta(m+4)}
$$
that is to say
$$
g(n):=\frac{(2^{2+n}-1)\zeta(n+2)}{(2^{4+n}-1)(2+n)(3+n)\zeta(n+4)}=g(m).
$$
Studying the variation of $g:(1,+\infty) \to \R$, we find that $g$ is strictly
decreasing, which means that there is no couple $(n,m)$ with $m<n$ such
that $g(m)=g(n)$.
This holds indeed because:
\begin{itemize}
\item $x\mapsto \displaystyle\frac{2^{2+n}-1}{2^{4+n}-1}$ is a strictly
decreasing function on $(1,+\infty)$;
\item $x\mapsto \displaystyle\frac{\zeta(x+2)}{(2+x)(3+x)\zeta(x+4)}$ is also
a strictly decreasing function on  $(1,+\infty)$, since the positive functions
$x\mapsto \displaystyle\frac{\zeta(x+2)}{\zeta(x+4)}$ and
$x\mapsto \displaystyle\frac{1}{(2+x)(3+x)} $ are strictly decreasing on
the same interval.
\end{itemize}
\end{proof}

We now give a result concerning the critical points of
$\Delta\mapsto  E^{\rm bip}_{nm}(A,\Delta)$ at fixed $A>0$.

\begin{prop}[\textbf{Bipartite phase and the asymptotics $A\to\infty$}]
The region of the bipartite phase is defined by $A>A_c^{nm}$:  
\begin{itemize}
\item[(1)]
for a given $A>A_c^{nm}$, there exists a nontrivial value of the parameter
$\Delta(A)\ne 1$ given by the equation
\begin{equation} \label{adel}
A^{m-n}=\frac{2^{n-m}[\zeta(m+1,1-\delta)-\zeta(m+1,\delta)]}{
\zeta(n+1,1-\delta)-\zeta(n+1,\delta)},\quad \delta=\frac{1}{1+\Delta}.
\end{equation}
\item[(2)] as $A\to +\infty$, we have
\begin{equation} \label{deltaas}
\Delta = 2A - 1 + o(1)
\end{equation}
regardless of the values of the LJ parameters $(n,m)$.
\end{itemize}
\end{prop}
\begin{proof}
The value of $\Delta$ is determined by the minimization condition for
the energy    
\begin{equation}
\frac{\partial}{\partial\Delta} E^{\rm bip}_{nm}(A,\Delta) = 0
\end{equation}
which gives
\begin{eqnarray}
\frac{1}{(2 A)^n} \left[ \zeta\left( n+1,\frac{1}{1+\Delta}\right) -
\zeta\left( n+1,\frac{\Delta}{1+\Delta}\right) \right] & & \nonumber \\
- \frac{1}{(2 A)^m} \left[ \zeta\left( m+1,\frac{1}{1+\Delta}\right) -
\zeta\left( m+1,\frac{\Delta}{1+\Delta}\right) \right] & = & 0 .   
\end{eqnarray}
It is clear that $\Delta=1$, which corresponds to the equidistant ground state,
always satisfies this equation; it provides the true energy minimum in the
region $0<A<A_c^{nm}$.
On the other hand, the trivial $\Delta=1$ corresponds to the energy maximum
in the bipartite region $A>A_c^{nm}$ where the nontrivial solution (\ref{adel})
with $\Delta>1$ provides the minimum of the energy.\\
We remark that in the limit $A\to +\infty$ it follows that, since $n>m$,
$$
\lim_{A\to +\infty} \frac{\zeta(m+1,1-\delta(A))-\zeta(m+1,\delta(A))}{
\zeta(n+1,1-\delta(A))-\zeta(n+1,\delta(A))}=0,\quad \delta(A)
=\frac{1}{1+\Delta(A)}
$$
which implies that $\displaystyle\lim_{A\to +\infty} \delta(A)=0$
(if not, the above limit cannot be $0$), and thus
$\displaystyle\lim_{A\to +\infty} \Delta(A)=+\infty$.
Therefore, in this $A\to +\infty$ regime, one can expand
the right-hand side of equation (\ref{adel}) in powers of small $\delta$
by using the expansion formulas
\begin{eqnarray} 
\zeta(m+1,\delta) = \frac{1}{\delta^{m+1}}+\zeta(m+1)-(m+1)\zeta(m+2)\delta
+ O\left( \delta^2\right) , \nonumber \\
\zeta(m+1,1-\delta) =  \zeta(m+1)+ (m+1) \zeta(m+1)\delta
+ O\left( \delta^2\right) , \label{zetaexp} 
\end{eqnarray}
implying
\begin{equation}
\zeta(m+1,1-\delta) - \zeta(m+1,\delta) = - \frac{1}{\delta^{m+1}}
+ 2 (m+1) \zeta(m+2) \delta + O\left( \delta^2\right) .
\end{equation}
The relation (\ref{adel}) thus yields
\begin{equation}
\left( \frac{2A}{1+\Delta} \right)^{m-n} = 1 +
O\left( \frac{1}{\Delta^{(3+m)}} \right)
\end{equation}  
which shows (\ref{deltaas}).
\end{proof}

\begin{figure}[tbp]
\begin{center}
\includegraphics[clip,width=0.5\textwidth]{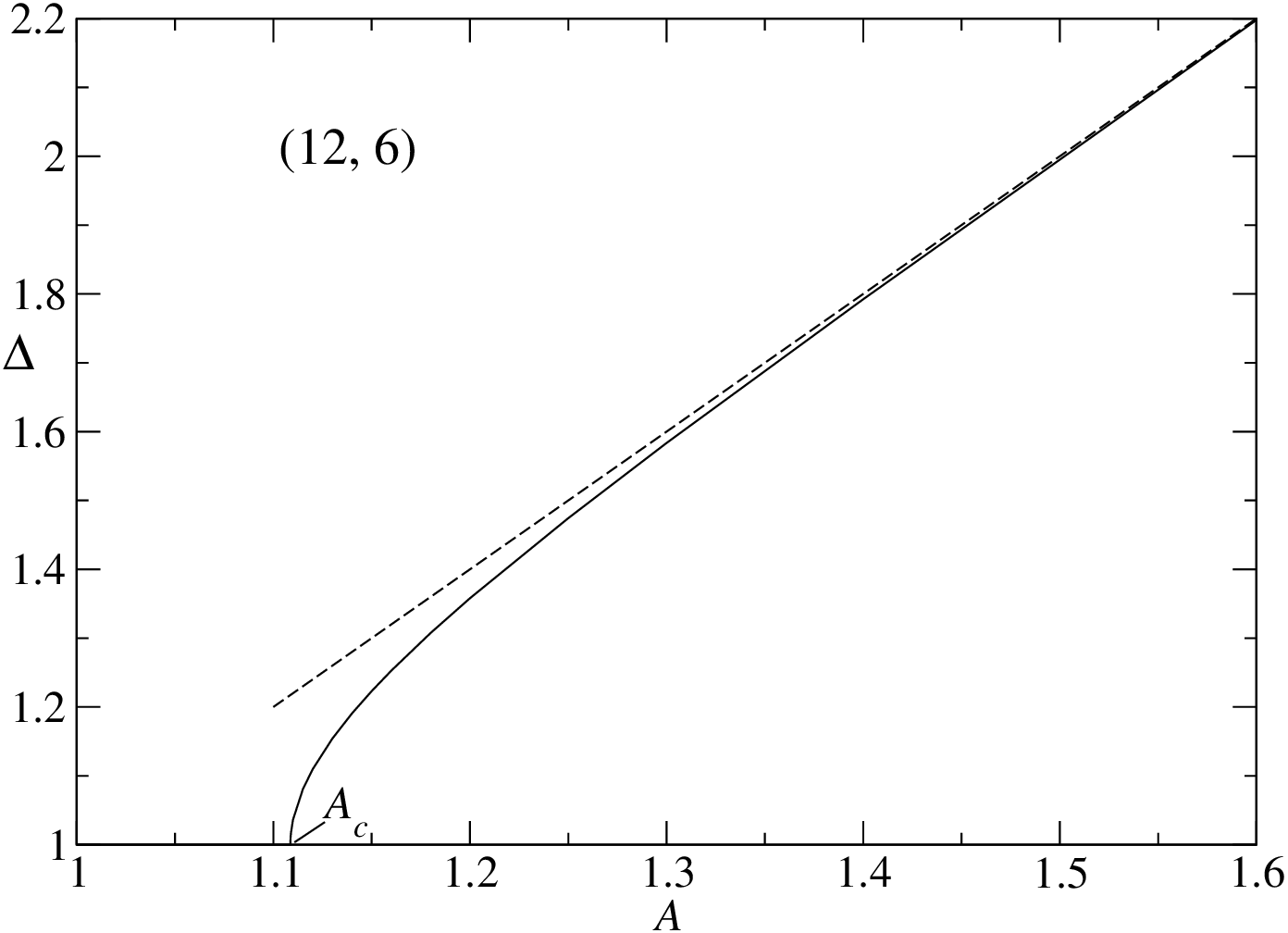}
\caption{The plot of the bipartite parameter $\Delta$ versus
$A$ for the 1D $(12,6)$ LJ model.
The equidistant value of $\Delta=1$ passes continuously into
the nontrivial one at the critical $A_c\equiv A_c^{12,6}$. 
The dashed line is the large-$A$ asymptote $\Delta\sim 2A-1$.}
\label{f5}
\end{center}
\end{figure}

\begin{example}[\textbf{Numerical checking of the critical points' asymptotics}]
For the LJ parameters $n=12$ and $m=6$, by using the relation (\ref{adel}),
the bipartite parameter $\Delta$ versus the inverse particle density $A$
is plotted in figure \ref{f5}.
The equidistant value of $\Delta=1$ passes continuously into
the nontrivial one at the critical $A_c^{12,6}$ given by (\ref{astar})
as follows
\begin{equation}
A_c^{12,6} = \frac{\pi}{2^{2/3}\sqrt{3}}
\left( \frac{5461}{6545} \right)^{1/6} \approx 1.10865478515.
\end{equation}
The dashed line is the large-$A$ asymptote $\Delta\sim 2A-1$
which is in fact the upper bound for $\Delta$ as the function of $A$.
Note a quick tendency of the plot $\Delta(A)$ to this asymptote for
relatively small values of $A$.
\end{example}

\renewcommand{\theequation}{4.\arabic{equation}}
\setcounter{equation}{0}

\section{Inclusion of the hard-core to 1D LJ model} \label{Sec4}
It was already mentioned above that in the limit $n\to\infty$ the Mie
interaction potential (\ref{Mie}) diverges for distances $r<1$,
thus creating a hard-core with radius 1.
Now we introduce for finite values of $n$ the true hard-core
with radius $\sigma>0$, defining the potential
\begin{equation} \label{fhc2}
f_{nm}^{\rm hc}(r) = \left\{
\begin{array}{ll}
+\infty & \mbox{if $r<\sigma$,} \cr
f_{nm}(r) & \mbox{if $r\ge \sigma$.}
\end{array}  
\right.  
\end{equation}
With respect to the symmetry $\Delta\leftrightarrow 1/\Delta$ of the bipartite
chain, we can restrict ourselves to $\Delta\ge 1$.
For a given bipartite chain, the smallest of alternating distances
(\ref{periods}), in our case $b(A,\Delta)$, must be greater than
$\sigma$, i.e. $\sigma\le \displaystyle\frac{2A}{1+\Delta}$.
The interval of admissible $\Delta$ values then becomes
\begin{equation} \label{deltain}
1 \le \Delta \le\frac{2A}{\sigma}-1 .
\end{equation}
These inequalities also imply the obvious condition
\begin{equation} \label{obm}
A\ge \sigma .
\end{equation}

\begin{itemize}
\item
If $\sigma\le 1$, the restriction (\ref{deltain}) is always fulfilled
and the added hard-core has virtually no influence on the dependence
$\Delta(A)$, except for the new condition $A\ge \sigma$.
\item
If $1\le \sigma \le A_c$, there are three possibilities for the values of $A$.
\\ (i)
For $\sigma \le A \le A_c$, we have the trivial solution $\Delta=1$ as in
the model without hard-cores.
\\ (ii)
For $A_c<A<A^*$, the plot of nontrivial $\Delta(A)$ follows the parabolic
trajectory given by equation (\ref{adel}), see also figure \ref{f5},
up to the point with the coordinates $(\Delta^*,A^*)$ given by
\begin{equation} \label{eq1}
\Delta^* = \frac{2A^*}{\sigma} - 1
\end{equation}
as this is the upper bound for $\Delta$ from the inequality (\ref{deltain}).
Note that the second relation for the coordinates $(\Delta^*,A^*)$
comes from (\ref{adel}) and reads as
\begin{equation} \label{eq2}
(A^*)^{m-n}=\frac{2^{n-m}[\zeta(m+1,1-\delta^*)-\zeta(m+1,\delta^*)]}{
\zeta(n+1,1-\delta^*)-\zeta(n+1,\delta^*)},
\end{equation}
where the notation $\delta^*=1/(1+\Delta^*)$ is made.
\\ (iii)
Finally, for $A>A^*$ the linear dependence
\begin{equation} \label{linear}
\Delta = \frac{2A}{\sigma} - 1
\end{equation}
is on, with $\Delta$ lying on the upper border of
the inequality (\ref{deltain}).
\item
If $\sigma\ge A_c$, the linear plot $\Delta=2A/\sigma-1$ starts
from $A=\sigma$, see the inequality (\ref{obm}).
\end{itemize}

\begin{figure}[tbp]
\begin{center}
\includegraphics[clip,width=0.5\textwidth]{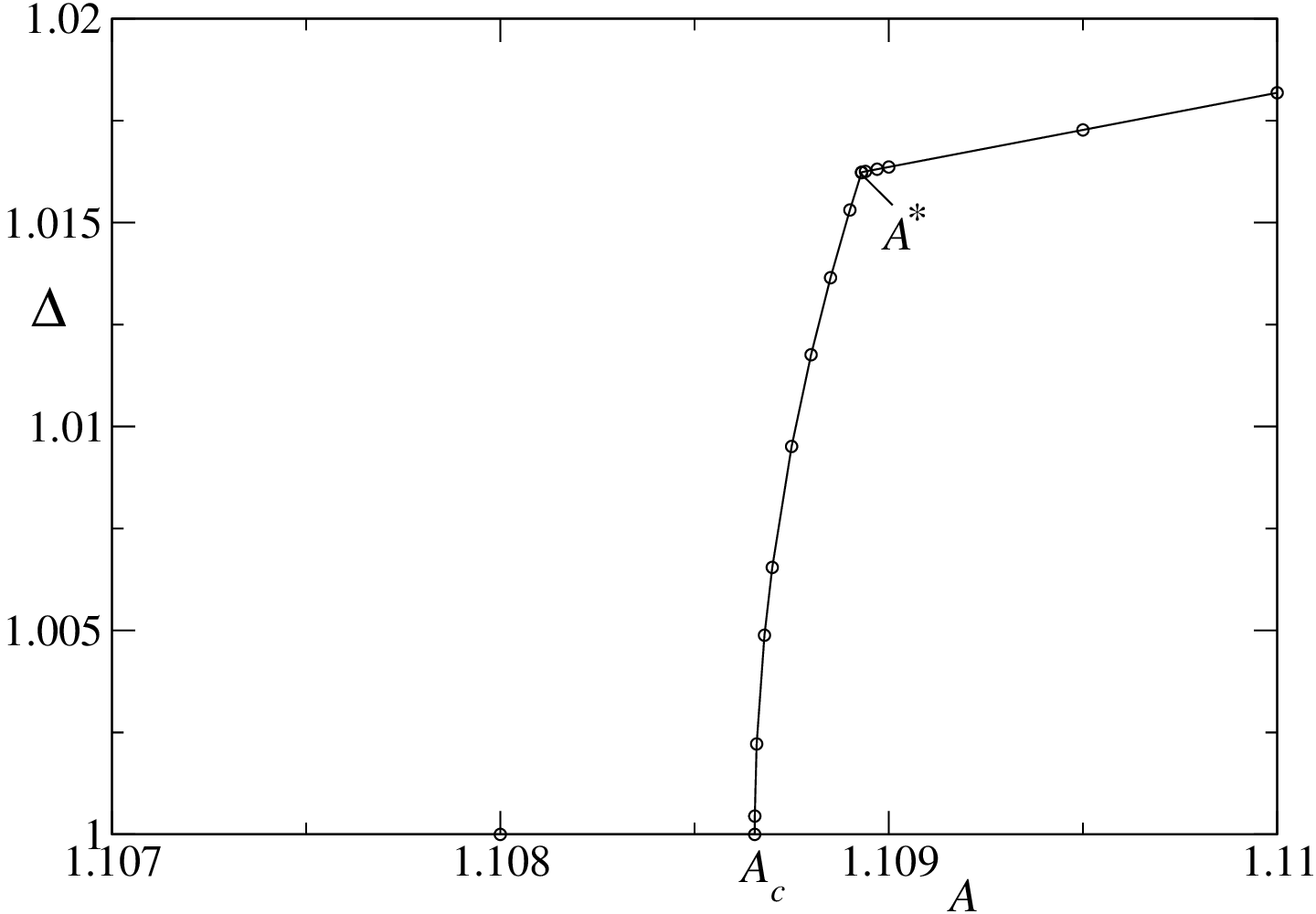}
\caption{The plot of the bipartite parameter $\Delta$ versus
$A$ for the 1D $(12,6)$ LJ model with the hard-core radius $\sigma=1.1$.
The plot is the same as the no-hard-core one in figure \ref{f5}
in the interval $\sigma<A<A^*$ with $A^*$ given by equations
(\ref{eq1}) and (\ref{eq2}).
For $A>A^*$, the linear dependence (\ref{linear}) is on.}
\label{f6}
\end{center}
\end{figure}

Let us take the hard-core radius $\sigma=1.1$ and perform the numerical
minimization for the standard $(12,6)$ LJ model whose dependence of
$\Delta$ on $A$ for the no-hard-core case is pictured in figure \ref{f5}.
The results, presented in figure \ref{f6}, illustrate the above scenario:
The dependence $\Delta(A)$ is the same as in figure \ref{f5} up to
the inverse particle density $A^*$ beyond which it is linear of type
(\ref{linear}).

It turns out that the value of $A^*$ diverges when $\sigma\to 1^+$.
\begin{prop}[\textbf{Asymptotic behavior of $A^*$ as $\sigma\to 1^+$}]
We have, as $\sigma\to 1^+$,
\begin{equation} \label{Delta}
A^*\sim\frac{1}{2} \left[ \frac{2(m+1)\zeta(m+2)}{(n-m)}\right]^{\frac{1}{m+2}}
(\sigma-1)^{-\frac{1}{m+2}}.
\end{equation}
\end{prop}
\begin{remark}[\textbf{Non-universality of the critical exponent}]
This means that the critical exponent $\tau=1/(m+2)$ is non-universal,
depending on the LJ parameter $m$, while the prefactor depends on both LJ
parameters $(n,m)$.
\end{remark}
\begin{proof}
We introduce another exponent $\tau$ such that $A^*$ is of order
$(\sigma-1)^{-\tau}$ as $\sigma\to 1^+$.
Let us now derive $\tau$ analytically.\\
It follows from equation (\ref{eq1}) that
\begin{equation} \label{deltastar}
\delta^* = \frac{1}{1+\Delta^*} = \frac{\sigma}{2A^*} ,
\end{equation}  
so that $\delta^*$ goes to 0 as $A^*\to \infty$.
Expressing from (\ref{deltastar}) $A^*=\sigma/(2\delta^*)$ and
using the small-$\delta^*$ expansions of the Hurwitz functions
(\ref{zetaexp}), the relation (\ref{eq2}) implies that
\begin{equation} \label{a0}
\sigma^{m-n} = \frac{1-2(m+1) \zeta(m+2){\delta^*}^{m+2}}{1-2(n+1)
\zeta(n+2){\delta^*}^{n+2}} +o(1), \qquad \delta^* \to 0 .
\end{equation}
As $\sigma-1$ is very small, we expand
$\sigma^{m-n}= 1+(m-n)(\sigma-1)+o(\sigma-1)$.
Since $n>m$, the term ${\delta^*}^{n+2}$ can be neglected compared to
${\delta^*}^{m+2}$ for infinitesimal $\delta^*$ in (\ref{a0}) and one gets
\begin{equation} \label{delta}
(n-m)(\sigma-1) = 2(m+1)\zeta(m+2) {\delta^*}^{m+2}+o(\sigma-1) ,
\qquad \sigma\to 1^+ .
\end{equation}
Using the relation (\ref{deltastar}), one gets the result.
\end{proof}

\begin{figure}[tbp]
\begin{center}
\includegraphics[clip,width=0.5\textwidth]{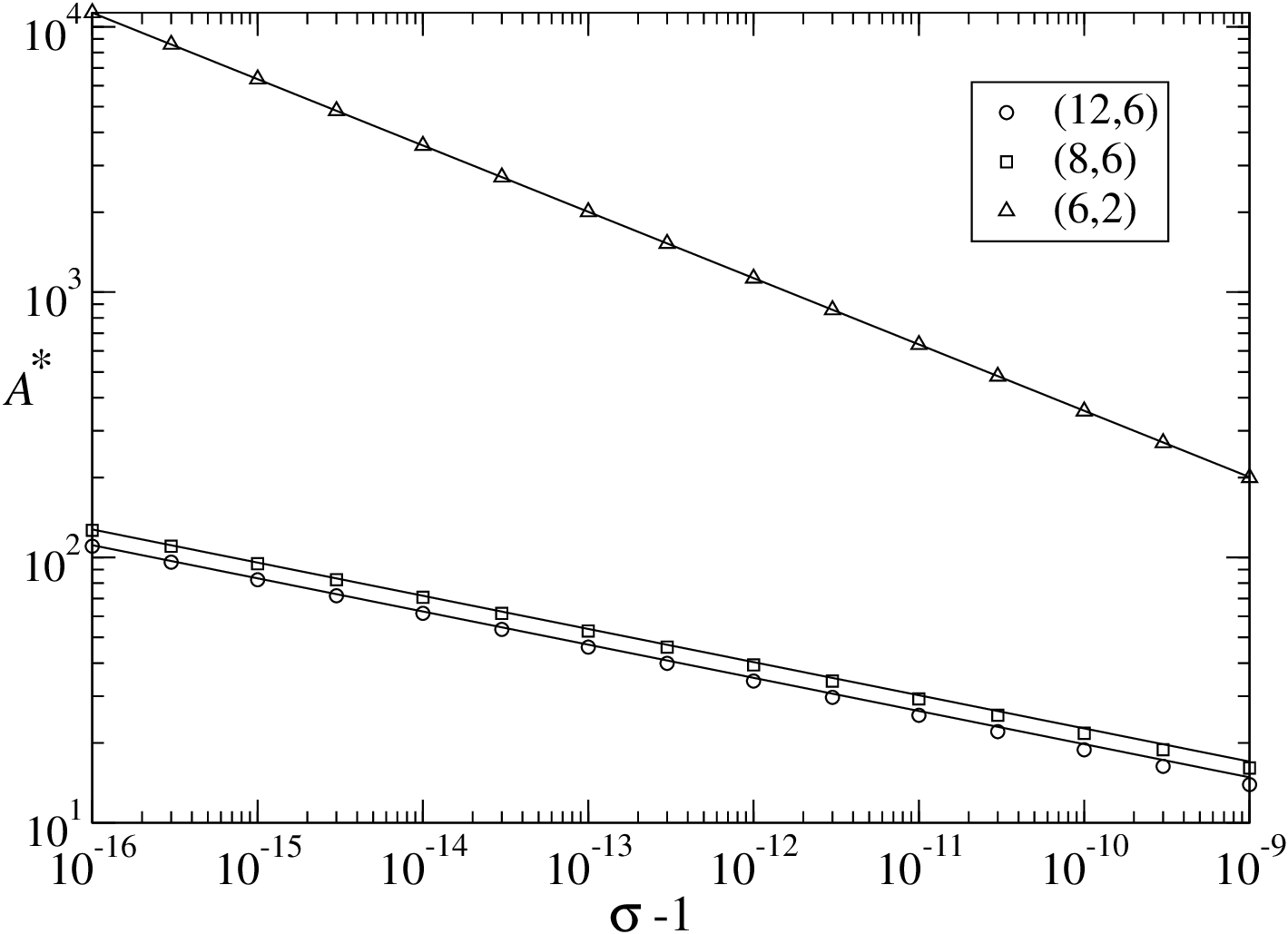}
\caption{The results for $A^*$ as the function of $\sigma-1$ for $\sigma$
very close to 1 from above, in the logarithmic scale.
Numerical data for the LJ parameters $(12,6)$ are represented by circles,
for $(8,6)$ by squares and for $(6,2)$ by triangles.
The corresponding linear plots deduced from the exact result (\ref{Delta})
are represented by solid lines.}
\label{f7}
\end{center}
\end{figure}

\begin{example}[\textbf{Numerical checking of the critical exponent}]
Numerical results for $A^*$ as the function of $\sigma-1$ for $\sigma$
very close to 1 from above, obtained by solving the set of equations
(\ref{eq1}) and (\ref{eq2}), are pictured in the logarithmic scale
in figure \ref{f7} for the LJ model with various $(n,m)$ parameters.
Numerical data for the LJ parameters $(12,6)$ are represented by circles,
for $(8,6)$ by squares and for $(6,2)$ by triangles.
The corresponding theoretical results deduced from (\ref{Delta}),
linear plots represented by solid lines, agree well with numerical data,
with tiny deviations for $\sigma-1$ of order or larger than $10^{-9}$.  
The slopes of data for the two cases with the same value of $m=6$,
namely $-0.12618$ for $(12,6)$ and $-0.12603$ for $(8,6)$, are close to
$-\frac{1}{8}$ predicted by the theory.
The corresponding prefactors from numerical fits are $1.10181$ and 
$1.26035$, agree well with the theoretical values $1.112290$ and
$1.276022$ obtained from (\ref{Delta}).
For the $(6,2)$ case we got $\tau=0.250026$, very close to the expected
$\frac{1}{4}$ and the prefactor $1.12763$ which is not far from
the predicted $1.128787$.
We tested the formula (\ref{Delta}) also for the $(6,3)$ LJ model,
the numerical value of $\tau=0.199995$ agrees with
the expected $\frac{1}{5}$.
\end{example}


%
%

\ack
The support received from the project EXSES APVV-20-0150 and the grant VEGA
No. 2/0089/24 is acknowledged. We also thank the anonymous referee for her/his useful suggestions and comments.

\end{document}